\documentclass[10pt,journal,twocolumn]{IEEEtran}

\usepackage[utf8]{inputenc}
\usepackage{ifthen}
\usepackage{amsthm}
\usepackage{amssymb,amsfonts}
\usepackage[ruled,linesnumbered]{algorithm2e}
\usepackage{cite}
\usepackage{graphicx}
\usepackage{tikz}
\usepackage{textcomp}
\usepackage{subfigure}
\usepackage{xcolor}
\usepackage{url}
\usepackage{threeparttable}
\usepackage{enumitem}
\usepackage{placeins}
\makeatletter
\providecommand{\subparagraph}{%
  \@startsection{subparagraph}{5}{\parindent}{3.25ex plus 1ex minus .2ex}%
  {-1em}{\normalfont\normalsize\bfseries}%
}
\makeatother
\usepackage{titlesec}

\usepackage{footnote}
\usepackage{multirow}
\usepackage{array}
\usepackage[T1]{fontenc}
\usepackage{setspace}
\usepackage[cmintegrals]{newtxmath}
\setlist[itemize]{topsep=0pt,partopsep=0pt,itemsep=0pt,parsep=0pt}
\setlist[enumerate]{topsep=0pt,partopsep=0pt,itemsep=0pt,parsep=0pt}
\titlespacing{\section}{0pt}{0.35ex plus 0.1ex minus 0.1ex}{0.25ex plus 0.1ex}
\titlespacing{\subsection}{0pt}{0.25ex plus 0.1ex minus 0.1ex}{0.15ex plus 0.1ex}
\titleformat{\subsubsection}[runin]{\normalfont\normalsize\itshape}{\arabic{subsubsection})}{0.5em}{}[.]
\titlespacing{\subsubsection}{0pt}{0.2ex plus 0.1ex minus 0.1ex}{0.1ex plus 0.1ex}
\AtBeginDocument{%
  \setlength{\abovedisplayskip}{2pt plus 1pt minus 1pt}%
  \setlength{\belowdisplayskip}{2pt plus 1pt minus 1pt}%
  \setlength{\abovedisplayshortskip}{1pt plus 0.5pt minus 0.5pt}%
  \setlength{\belowdisplayshortskip}{1pt plus 0.5pt minus 0.5pt}%
  \setlength{\itemsep}{0pt plus 0ex}%
}

\newcommand{\bm}[1]{\mbox{\boldmath $#1$}}

\newtheorem{lemma}{Lemma}

\newtheorem{thm}{Theorem}

\newtheorem{remark}[thm]{Remark}

\ifCLASSINFOpdf

\else

\fi

\begin{document}

\title{
Graph Attention-based Decentralized Actor-Critic for Dual-Objective Control of Multi-UAV Networks
}

\urldef{\gadcDataUrl}\url{https://github.com/Haoran-Peng/GADC}
\author{Haoran~Peng,~\IEEEmembership{Member,~IEEE}, and Ying-Jun Angela Zhang,~\IEEEmembership{Fellow,~IEEE}
\thanks{
\protect\setlength{\baselineskip}{8.1pt}%
This work is supported in part by the General Research Fund (project number 14202723, 14207624), Area of Excellence Scheme grant (project number AoE/E-601/22-R), Collaborative Research Fund (C5101-25YF), and NSFC/RGC Collaborative Research Scheme (project number CRS\_HKUST603/22, CRS\_HKU702/24), all from the Research Grants Council of Hong Kong.
An earlier version of this paper was presented in part at the 2025 IEEE Global Communications Conference Workshops (GC Wkshps)~\cite{GC2025}.
(Corresponding author: Ying-Jun Angela Zhang.)

Haoran Peng and Ying-Jun Angela Zhang are with the Department of Information Engineering, The Chinese University of Hong Kong, Shatin, N.T., Hong Kong SAR (Email:~hrpeng@ie.cuhk.edu.hk; yjzhang@ie.cuhk.edu.hk). 
\newline Data is available on-line at \gadcDataUrl
}
}

\markboth{}%
{Peng and Zhang: Graph Attention-based Decentralized Actor-Critic for Multi-UAV Networks}
\maketitle

\begin{abstract}
We study dual-objective control in multi-UAV systems, where service coverage is the primary goal and battery lifetime is the secondary goal.
We propose a Graph Attention-based Decentralized Actor-Critic (GADC) to optimize the dual objectives. 
The proposed approach employs a graph attention network (GAT) to aggregate
information from a variable number of in-range neighboring UAVs into a
fixed-dimensional topology-aware representation for decentralized decision making.
We then develop a dual-critic network to handle coverage improvement and lifetime control in the same decentralized policy.
Within this dual-critic framework, the secondary energy objective uses a replay-based clipped actor update with a fixed post-primary reference. A conditional lower bound on coverage change is derived for candidate energy-refined policies satisfying a KL-neighborhood condition.
We analyze the computational complexity of GADC and evaluate it against representative baselines on held-out ideal layouts and deterministic urban channels generated with NVIDIA Sionna ray tracing.
The results show competitive coverage and lifetime performance. 
In the nominal-density sensitivity study, coverage varies less across the tested clipping parameters than across the tested weighted-sum coefficients.
\end{abstract}


\section{Introduction}
\label{Sec: Intro}
\IEEEPARstart{C}onventional unmanned aerial vehicle (UAV) systems rely on direct connections between UAVs and ground user terminals (UTs)~\cite{10820057}. 
This architecture is inherently limited by the communication range and battery capacity of the UAVs.
The limited communication range results in localized observation and non-uniform service coverage~\cite{10834514}. 
UAVs handling heavy service loads experience rapid energy depletion, significantly reducing the multi-UAV network's operational lifetime, defined as the duration until the first UAV exhausts its power. 
This motivates a multi-UAV navigation strategy that treats service coverage as the primary objective and UAV network lifetime as the secondary objective~\cite{9697395,11031135}.
In particular, maximum coverage navigation aims to maximize the number of UTs covered across all time steps. 
Likewise, maximum lifetime navigation aims to delay the first UAV battery depletion by accounting for UAV movement, UAV-UAV communications, and UAV-UT service loads~\cite{9687317}.

Optimizing the dual objectives, maximizing service coverage while extending UAV network lifetime, requires careful management of their inherent trade-offs~\cite{9763515}. 
Meanwhile, the decentralized nature of UAVs makes global information acquisition across the multi-UAV network difficult, complicating the development of a collaborative navigation strategy among all UAVs~\cite{10972017}.

Recent works use Graph Attention Network (GAT) layers and message passing to help UAVs aggregate multi-hop neighbor observations in partially observable environments~\cite{9697395}. These UAV-UT interactions are commonly modeled as a Markov Decision Process (MDP) and solved with Multi-Agent Reinforcement Learning (MARL) for decentralized navigation~\cite{9892688}. However, the joint control of service coverage and network lifetime is not a generic multi-objective RL task. The coverage objective is dense and immediate, since it depends on how many UTs are served at each time slot. In contrast, network lifetime is determined by the time to the first battery depletion. Minimum residual energy provides a training surrogate for this stopping-time objective, but maximizing its discounted return is not equivalent to maximizing lifetime. The two objectives become strongly conflicting when the energy budget is tight, because improving coverage requires UAV movement and service transmission, both of which consume battery energy. Decentralized execution adds another difficulty, as each UAV only observes nearby UTs and in-range neighboring UAVs.
Existing methods often use weighted-sum scalarization to combine the two objectives~\cite{10420488}. Such scalarization is sensitive to reward scaling and can suffer from gradient dominance. The relative reward scales and coefficients therefore influence the learned coverage-energy tradeoff.
\subsection{State of the Art}
\label{subsec:sota}
Existing studies address UAV-assisted network control through connectivity optimization, multi-UAV control and navigation, and multi-objective RL.
\subsubsection{UAV-assisted Cellular Connectivity}\hspace{0.2em}
UAVs can be rapidly deployed to extend wireless coverage in infrastructure-limited and ad hoc environments~\cite{10214219}. Recent studies have integrated UAV connectivity with Mobile Edge Computing (MEC), energy harvesting, trajectory design, resource allocation, UT-satellite relaying, and secure uplink transmission to improve service quality under energy or power constraints~\cite{10460374,9814972,10438999,10243608}. However, limited onboard energy and local observation still restrict long-duration cooperative operation. In decentralized multi-UAV networks, each UAV only observes nearby UTs and neighboring UAVs, which makes global coordination difficult.

{
UAV-assisted MEC introduces additional dependencies among mobility, computation, and energy availability.
\cite{chen2023triple} used three coupled learners for trajectory planning, energy replenishment, and application placement to improve long-term computing energy efficiency.
\cite{sun2024multiobjective} jointly optimized task offloading, computing resources, and UAV trajectories to balance completion delay, UAV energy consumption, and offloaded task volume.
These studies motivate coordination that accounts for both spatial service demand and onboard resources. 
GADC focuses on decentralized navigation with coverage as the primary objective and network lifetime as the secondary objective.
Its two critics evaluate different objectives for a shared navigation policy, whereas the triple-learner approach separates different scheduling decisions.
}

\subsubsection{Multi-UAV Control and Navigation}\hspace{0.2em}
Early multi-UAV studies mainly relied on centralized control for swarm management, secure UAV mesh networking, image acquisition, target clustering, bandwidth allocation, and trajectory or computation-resource optimization under budget, safety, {quality-of-experience (QoE)}, and energy constraints~\cite{9446301,9437802,10124817,10834514,10934003}. Centralized RL has also been used to balance coverage and energy efficiency~\cite{8432464}, but its dependence on global information or fixed infrastructure limits its applicability in ad hoc multi-UAV networks.

Decentralized MARL enables cooperative navigation from local observations~\cite{10254323}, but standard MARL remains vulnerable to partial observability. Networked PPO and GAT-based MARL improve information sharing through communication-aware learning and neighbor aggregation~\cite{10591960,9697395}. However, these methods mainly address single-objective control and do not explicitly model the conflict between service coverage and network lifetime.
\subsubsection{Multi-Objective and Constrained RL}\hspace{0.2em}
{
Conflicting objectives can be handled by weighted scalarization, Pareto search, or constrained policy optimization~\cite{10.1145/3514233,pmlr-v70-achiam17a}.
Weighted scalarization depends on reward scales and coefficient selection. Pareto search seeks alternative tradeoff solutions from which an operating point can be selected. The software-engineering comparison in~\cite{10.1145/3514233} concerns these search strategies rather than UAV coordination.
Constrained policy optimization (CPO) incorporates an explicit policy constraint with a local constrained update~\cite{pmlr-v70-achiam17a}.
The considered mission calls for one decentralized policy with coverage prioritized over energy refinement, without prescribing a fixed minimum coverage or lifetime threshold. Accordingly, GADC maintains one shared actor and assigns coverage and energy different update roles. A coverage-oriented step is followed by a replay-based clipped energy refinement with a fixed post-primary reference. This update ordering expresses the mission priority without establishing exact lexicographic optimality or enforcing a prescribed KL bound. The KL threshold used in the analysis defines a policy neighborhood, not a mission-performance requirement.}

Table~\ref{tab:relatedwork} compares the information structure and objective handling of related approaches. GADC combines graph-based decentralized representation with sequential dual-critic actor updates. A KL-constrained policy formulation provides a reference for the conditional performance analysis. 
{
Unlike the heuristic dual-objective formulation in our preliminary study~\cite{GC2025}, the primary coverage action-value critic guides a discrete-action policy update, while the secondary energy state-value critic supplies an energy-reward advantage for replay-based clipped refinement inspired by PPO~\cite{DBLP:journals/corr/SchulmanWDRK17}.
The clipped surrogate uses a fixed post-primary reference. The coverage-change analysis applies to candidate policies satisfying the KL-neighborhood assumption in Lemma~\ref{lemma1}, not to every clipped update.
}

\subsection{Contributions}
This study investigates dual-objective control in decentralized multi-UAV networks, where service coverage is the primary objective and network lifetime is the secondary objective. 
The main contributions are summarized as follows.
\begin{itemize}
\item {\textbf{Primary-secondary dual-critic optimization.}
GADC uses a shared actor with separate coverage action-value and energy state-value critics. A coverage-oriented actor step is followed by an energy-reward refinement using a fixed post-primary reference. The objectives act through sequential updates rather than a weighted sum of their rewards.}
\item \textbf{Sensitivity of the clipped secondary refinement.}
{
At the nominal density, the coverage range under the weighted-sum coefficient $\varphi$ is about $7.2$ times that under $\epsilon$.
Cross-density retraining with 20, 30, and 40 UAVs further shows that $\epsilon=0.2$ remains within the variation across independently trained policies over the tested twofold density range.
}
{
\item \textbf{Conditional performance analysis.}
Under a KL-neighborhood condition between the coverage-oriented policy
and a candidate energy-refined policy, we derive a lower bound on the
resulting change in coverage return.
}
{
\item \textbf{Graph-based decentralized representation.}
The GAT-GRU encoder aggregates information from a variable number of in-range
neighboring UAVs into a fixed-dimensional representation and captures temporal
information through the GRU. The shared encoder and adjacency-masked aggregation
support decentralized execution across different fleet sizes.
}

\item \textbf{Empirical validation.}
We conduct controlled ablation studies to isolate the effects of the coverage action-value critic, the GAT encoder, and the clipped secondary update. We also test sensitivity to the clipping parameter, scalability with the numbers of UAVs and UTs, topology refresh interval, sensor noise, and non-stationary user mobility. These evaluations use held-out layouts and ray-tracing digital-twin deployments under deterministic urban channels.
\end{itemize}

{
Graph aggregation and recurrent memory address partial observability, whereas the dual-critic update determines how coverage and energy affect the policy.
GADC uses one shared actor, a coverage action-value critic, and an energy state-value critic. The actor receives a coverage-oriented update followed by a replay-based clipped energy refinement.
The ablation and sensitivity studies examine these components empirically, while Lemma~\ref{lemma1} characterizes coverage change under its stated KL-neighborhood condition.
}

The rest of the paper is structured as follows.
Section~\ref{Sec: System Model} presents the system model.
The system operation and the proposed GADC are detailed in Section~\ref{Sec:DualAC}.
Section~\ref{numerical result} gives numerical results.
Finally, Section~\ref{Conclusions} concludes the paper.

\begin{table*}[!tbp]
\centering
\caption{{Comparison of information structure and objective handling}}
\label{tab:relatedwork}
\begingroup

\footnotesize
\setlength{\tabcolsep}{3pt}
\renewcommand{\arraystretch}{1.12}
\begin{tabular}{p{.075\textwidth}|p{.13\textwidth}|p{.19\textwidth}|p{.18\textwidth}|p{.17\textwidth}|p{.16\textwidth}}
\hline\hline
Reference & Task & Information and representation & Objective handling & Policy and value organization & Output and analytical scope\\
\hline
\cite{8432464} & UAV coverage & Centralized network state & Composite energy-efficiency reward & DDPG actor and action-value critic & Continuous control policy\\
\hline
\cite{9697395} & Partially observed UAV coverage & Local observations, graph aggregation, and recurrent memory & Local heuristic reward and maximum-entropy learning & Graph-recurrent soft value learning & Stochastic navigation policy\\
\hline
\cite{pmlr-v70-achiam17a} & Constrained RL & State-based policy representation & Explicit constrained policy optimization & Policy update with reward and constraint estimates & Approximate constrained update with theoretical assumptions\\
\hline
GAT-MADDPG ablation & Coverage and energy & Same GAT-GRU encoder as GADC & Weighted coverage-energy reward & Shared actor and one weighted-reward critic & One navigation policy\\
\hline
GADC & Coverage and first-depletion lifetime & Local observations with GAT-GRU neighbor aggregation & Sequential coverage and clipped energy updates & Shared actor with coverage $Q^c$ and energy $V^f$ critics & One policy and conditional coverage-change analysis\\
\hline\hline
\end{tabular}
\endgroup
\end{table*}

\begin{figure*}[t]
	\centering
	\subfigure[] {
    	\includegraphics[width=.475\textwidth]{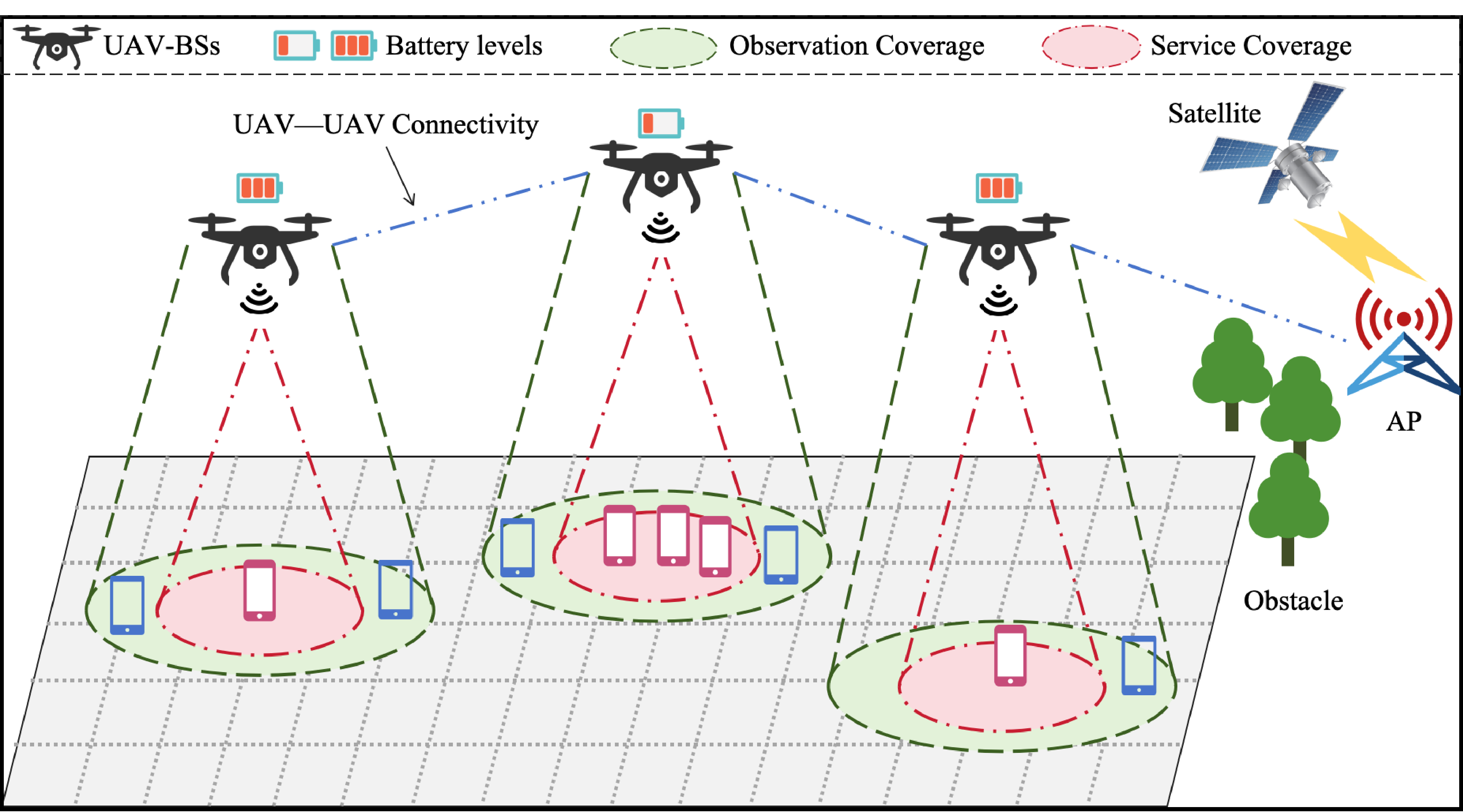}
        	\label{idealarea}
	 }
	 \subfigure[] {
        	\includegraphics[width=.475\textwidth]{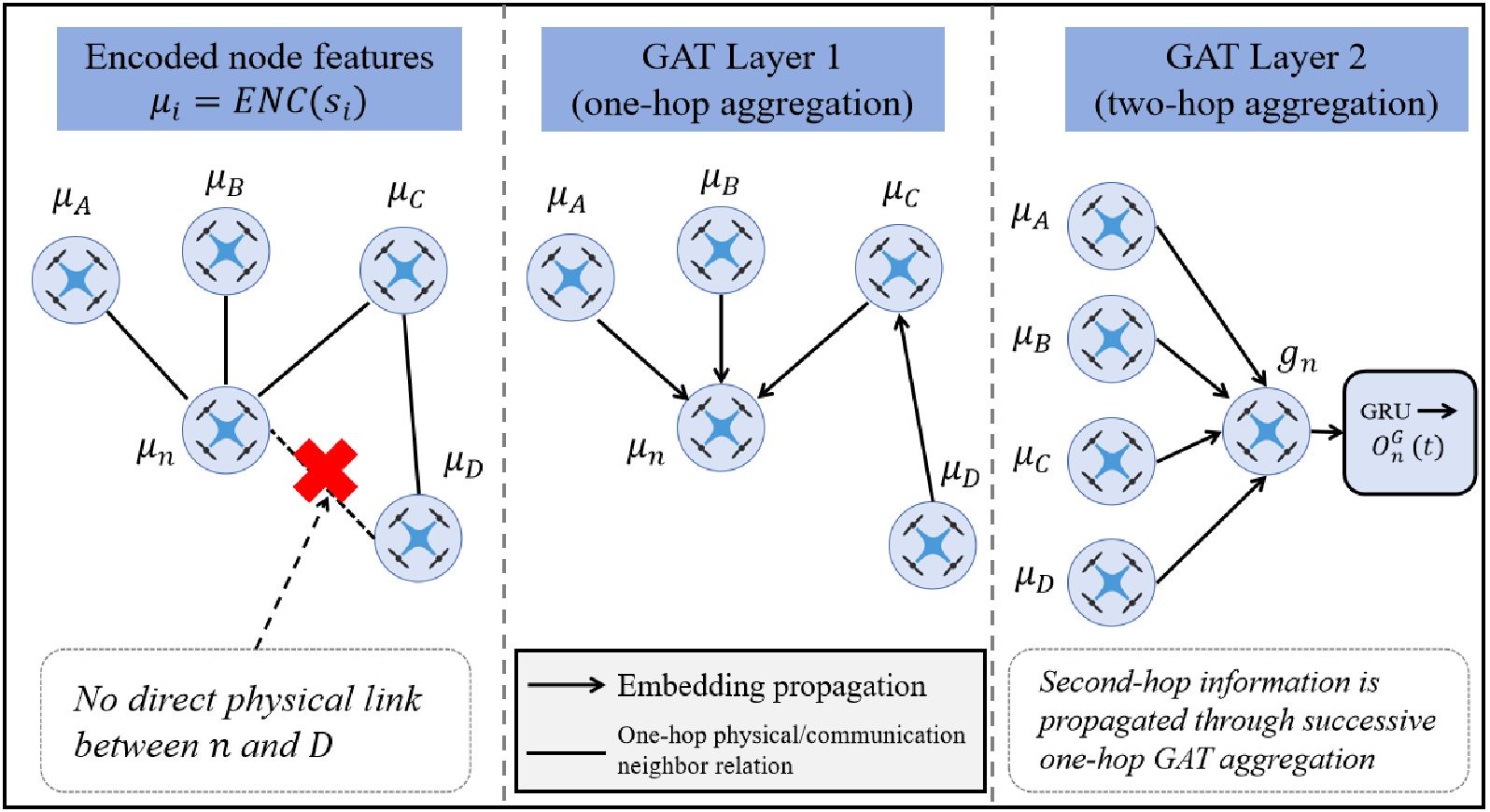}
        	\label{fig:gat_aggregation_overview}
	 }
	 \caption{
	 {Physical connectivity and graph-information propagation in GADC.
     (a) Physical service, observation, and one-hop UAV-UAV communication relationships.
     (b) Two-layer adjacency-masked GAT aggregation. At each layer, a UAV aggregates encoded features only from its one-hop communication neighbors. Stacking two GAT layers allows information originating from two-hop neighbors to influence the local representation through intermediate UAVs. The arrows in (b) represent embedding propagation rather than additional physical wireless links.}
	 }
  \label{scenario_setting}
\end{figure*}

\section{System Model}
\label{Sec: System Model}
{Suppose that $M$ UTs and $N$ UAVs are randomly distributed in a horizontal area of $E\,\mathrm{m}\times E\,\mathrm{m}$.}
Each UAV and UT is equipped with a single antenna. 
The heights of the UAV and UT antennas are fixed.
{In the nominal setting, UT locations remain unchanged within an episode and vary between episodes. Section~\ref{subsec:robustness} additionally evaluates mobile UTs.}
Suppose that UAVs function as flying BSs to serve the UTs and are navigated according to the locations of UTs.
The system maintains two types of communications: UAV-UAV links for information exchange between neighboring UAVs and UAV-UT links for service provisioning.
Table~\ref{tab:Notations} summarizes the key notations used throughout this paper.
\begin{table}[!t]
    \centering
    \caption{Glossary of Notations}
    \begin{tabular}{c|l}
    \hline \hline
      \bfseries      \textbf{Symbol} & \textbf{Description}   \\ \hline \hline
      $M$ & The number of UTs \\\hline
      $N$ & The number of UAVs  \\\hline
        {$D_{\mathrm{svc}}$} & {The UAV-UT service-threshold distance} \\\hline
        {$D_s$} & {The UAV-UAV communication range} \\\hline
        $R_{s}$ & The maximum horizontal service radius\\\hline
        {$R_o$} & {The horizontal observation radius}\\\hline
        {$H$} & {The vertical UAV-UT separation} \\\hline
        $T$ & The total number of time slots \\\hline
        $c_{n}(t)$ & The number of UTs served by UAV $n$ \\& at time slot $t$ \\\hline
        $\omega_{m,n}^{s}(t)$ & The service coverage indicator of the UT $m$\\ &and the UAV $n$ \\\hline
        {$G(\mathcal V,L)$} & {The UAV communication graph}\\\hline
        {$L$} & {The set of UAV-UAV edges} \\ \hline
        {$\mathcal{A}(t)$} & {The time-varying adjacency matrix} \\\hline
        {$\mathcal{V}$} & {The set of UAV nodes} \\\hline
        {$\mathcal I_n(t)$} & {The one-hop neighbor set of UAV $n$} \\\hline
        $\Delta t$ & The duration of a time slot \\\hline
        $p_{n}^{s}$ & The transmit power for the link from \\&UAV $n$ to any UTs it serves\\\hline
        $\rho_{n,i}$ & {The fixed effective coordination energy} \\&{per slot for neighboring link $n\to i$}\\\hline
        $\chi_{h}$ & The hovering energy consumption for a UAV \\\hline
        $b_{n}(t)$ & The residual energy of the UAV $n$ \\&at the end of time slot $t$\\\hline
        $T_{net}$ & The network lifetime of the UAV-BS system \\\hline
        $h_{n}(t)$ & The hidden state for UAV $n$ at the \\ & time slot $t$ in the memory unit \\\hline
        {$g_n$} & {The multihead aggregated representation} \\ & {of UAV $n$} \\\hline
        {$O_n^G$} & {The final GAT-GRU representation of UAV $n$} \\\hline
        $d^\pi(s)$ & Discounted state visitation distribution under $\pi$ \\ \hline
$\pi^{c}, \pi^{f}$ & Policies updated by primary/secondary critics \\ \hline
$\epsilon$ &  {Likelihood-ratio clipping parameter} \\ \hline
$\delta$ & {KL-neighborhood bound} \\ \hline
$\gamma$ & Discount factor \\ \hline
    \end{tabular}
    \label{tab:Notations}
    \vspace{-3pt}
\end{table}

\subsection{UAV Observation, Service, and Communication Ranges}
{
The ideal simulator uses a distance-based coverage abstraction~\cite{6863654} with fixed horizontal observation and service radii $R_o$ and $R_s$, as illustrated in Fig.~\ref{idealarea}. It does not simulate instantaneous Rayleigh fading. These radii represent received-power thresholds $\sigma_o$ and $\sigma_s$, respectively, with $\sigma_s\geq\sigma_o$.

Let $D_{\mathrm{svc}}$ denote the UAV-UT slant distance at the service threshold and let $H$ denote the vertical separation between a UAV and a UT. In the ideal geometric model, we assume a common maximum direct-link range for UAV-UT service and UAV-UAV communication, such that $D_{\mathrm{svc}}=D_s$. Since the UAVs operate at the same altitude, $D_s$ is also their maximum horizontal separation for a direct communication link. The corresponding horizontal service radius is
}
\begin{equation}
    \begin{aligned}
        {R_s=\sqrt{D_{\mathrm{svc}}^2-H^2}=\sqrt{D_s^2-H^2}}.
    \end{aligned}
    \label{eq:service_radius}
\end{equation}
{The common-range assumption is used only for this ideal geometry. In the urban digital-twin simulator, buildings and other obstacles affect UAV-UT propagation. Serviceability is determined by the ray-traced SNR threshold, which corresponds to a received-power threshold for the fixed noise level, rather than by a fixed spherical service range. The UAV-UAV communication range is specified separately.}

{
Each UT is assigned to at most one eligible UAV. A UT without an eligible UAV remains unserved.
The selection rule is specified for each simulator in Section~\ref{numerical result}, with ties resolved by the smallest UAV index.
There is no explicit per-UAV limit on the number of served UTs. Serving more UTs increases the service-energy cost.
}
The number of UTs served by the $n^{th}$ UAV during time slot $t$ is:
\begin{equation}
    \begin{aligned}
    c_{n}(t) = \sum_{m=1}^{M}\omega_{m,n}^{s}(t),
    \end{aligned}
\end{equation}
where $\omega_{m,n}^{s}(t) = 1$ indicates that UT $m$ is served by the $n^{th}$ UAV during time slot $t$ and $\omega_{m,n}^{s}(t) = 0$ otherwise.
$\sum_{n=1}^{N}\omega_{m,n}^{s}(t)\leq 1$, $\forall m$ indicates that a UT is served by at most one UAV during a time slot.

{
To facilitate decentralized cooperation, UAVs exchange information with
neighbors within the air-to-air communication range $D_s$. Under the
assumed LoS model, a link between UAVs $n$ and $i$ is available when
$d_{n,i}(t)\leq D_s$. The numerical model assumes synchronized exchanges
over orthogonal control resources at each slot boundary.
A practical implementation can use a common global navigation satellite
system (GNSS) time reference and a time-division multiple access (TDMA)
schedule assigned before deployment. Each UAV is assigned a unique
transmission minislot across the fleet, without spatial reuse, and
transmits once in each aggregation round. This allocation prevents
collisions among scheduled UAV transmissions without requiring a
centralized navigation controller. Each packet contains the sender
identifier, slot and aggregation-round indices, position, residual
energy, and the embedding required by the corresponding GAT layer.
The position information received in the first round is used to construct
the dynamic adjacency matrix $\mathcal{A}(t)$ before the movement decision.
For the two-layer GAT, the first round exchanges local embeddings and the
second round exchanges the first-layer outputs. Each attention head is
then computed locally from the received embeddings.
}

{
Let $\tau_{\mathrm{pkt}}$ denote the maximum packet airtime,
$\tau_{\mathrm{g}}$ the guard interval, and $T_{\mathrm{proc}}$ and
$T_{\mathrm{sync}}$ the processing and synchronization budgets,
respectively. With $L_{\mathrm{GAT}}=2$ aggregation rounds, the reserved
control interval must satisfy
\begin{equation}
L_{\mathrm{GAT}}N
\left(\tau_{\mathrm{pkt}}+\tau_{\mathrm{g}}\right)
+T_{\mathrm{proc}}+T_{\mathrm{sync}}
\leq T_{\mathrm{ctrl}},
\qquad
T_{\mathrm{ctrl}}\ll\Delta t.
\label{eq:control_interval_budget}
\end{equation}
The guard interval must accommodate synchronization error and propagation-delay variation. In this deployment design, a late packet would be excluded from the corresponding aggregation round, while the local embedding would remain available through the fusion path.
The control-interval inequality in Eq.~\eqref{eq:control_interval_budget} is a deployment feasibility condition. Packet airtime, clock error, queueing, retransmission, and packet-level contention are abstracted in the current simulator, which therefore does not report a numerical control-channel latency.
The coordination term in Eq.~\eqref{RESIDUALEE} is an effective energy parameter rather than a cost derived from measured packet exchanges.
}

\subsection{UAV Energy Consumption Model}
\label{Sec:EnergyModel}
The energy consumption of each UAV is modeled by considering four operational components: navigation, ground service, inter-UAV communication, and hovering.
First, the \textbf{navigation energy}, denoted by $E_n^{nav}(t)$, captures the motion-dependent energy required for horizontal displacement. 
While the aerodynamic power of rotary-wing UAVs is a non-linear function of velocity, typically involving cubic terms regarding blade profile and induced power~\cite{7888557}, {we adopt a linear approximation to keep the training model tractable.}
Assuming the UAVs operate primarily within an energy-efficient velocity regime, the navigation energy is modeled as proportional to the traveled distance~\cite{9697395}:
\begin{equation}
    E_n^{nav}(t) = \kappa \eta_{n}(t),
\end{equation}
where $\eta_{n}(t)$ denotes the horizontal displacement of UAV $n$ during time slot $t$, and $\kappa$ represents the energy consumption coefficient per meter.
Second, the \textbf{service energy} required for communicating with ground UTs is given by $E_n^{ser}(t) = \Delta t c_{n}(t)p_{n}^{s}$, which scales with the time slot duration $\Delta t$, the number of served UTs $c_{n}(t)$, and the constant transmit power $p_{n}^{s}$.
{
Third, the coordination energy
$\sum_{i\in\mathcal{I}_n(t)}\rho_{n,i}$ uses a fixed coefficient
$\rho_{n,i}$ for each neighboring link per slot. The total coordination
cost varies with the neighbor set $\mathcal{I}_n(t)$. This term represents
an effective coordination overhead rather than the measured energy
of individual packets or GAT exchange rounds. Finally, each UAV
incurs a fixed static energy cost $\chi_h$ per slot, which represents
hovering and a baseline allowance for onboard electronics. The
computational energy of GAT inference is not measured separately.
}
Combining these components, the residual battery energy $b_n(t)$ of UAV $n$ at the end of time slot $t$ evolves as
{
\begin{equation}
\label{RESIDUALEE}
b_n(t) = \Bigg[b_n(t-1) - \underbrace{\kappa \eta_{n}(t)}_{\text{Navigation}}-\underbrace{\Delta t c_{n}(t)p_{n}^{s}}_{\text{Service}}-\underbrace{\sum_{i\in \mathcal{I}_{n}(t)}\rho_{n,i}}_{\text{Coordination}}-\underbrace{\chi_{h}}_{\text{Static}}\Bigg]^+,
\end{equation}
}
where $[x]^+=\max\{x,0\}$. $b_{n}(t-1)$ denotes the residual energy from the previous slot. 
{
The action selected at time slot $t$ is executed before the battery state is updated. 
If $b_n(t)=0$ for any UAV after this update, the corresponding transition is terminal and the episode ends.
This clipping and termination rule is a terminal-state convention. It does not check whether the remaining battery can support the selected action and its service and communication costs before execution. No energy-based action mask is applied.
Accordingly, the network lifetime is
\begin{equation}
T_{\mathrm{net}}
=
\min\left\{
t\geq1:\min_{n\in\{1,\ldots,N\}} b_n(t)=0
\right\}.
\label{eq:network_lifetime}
\end{equation}
The depletion slot is included in this count. If no depletion occurs, $T_{\mathrm{net}}=\infty$.
}

\begin{figure*}[!tbp]
\centerline{\includegraphics[width=.95\textwidth]{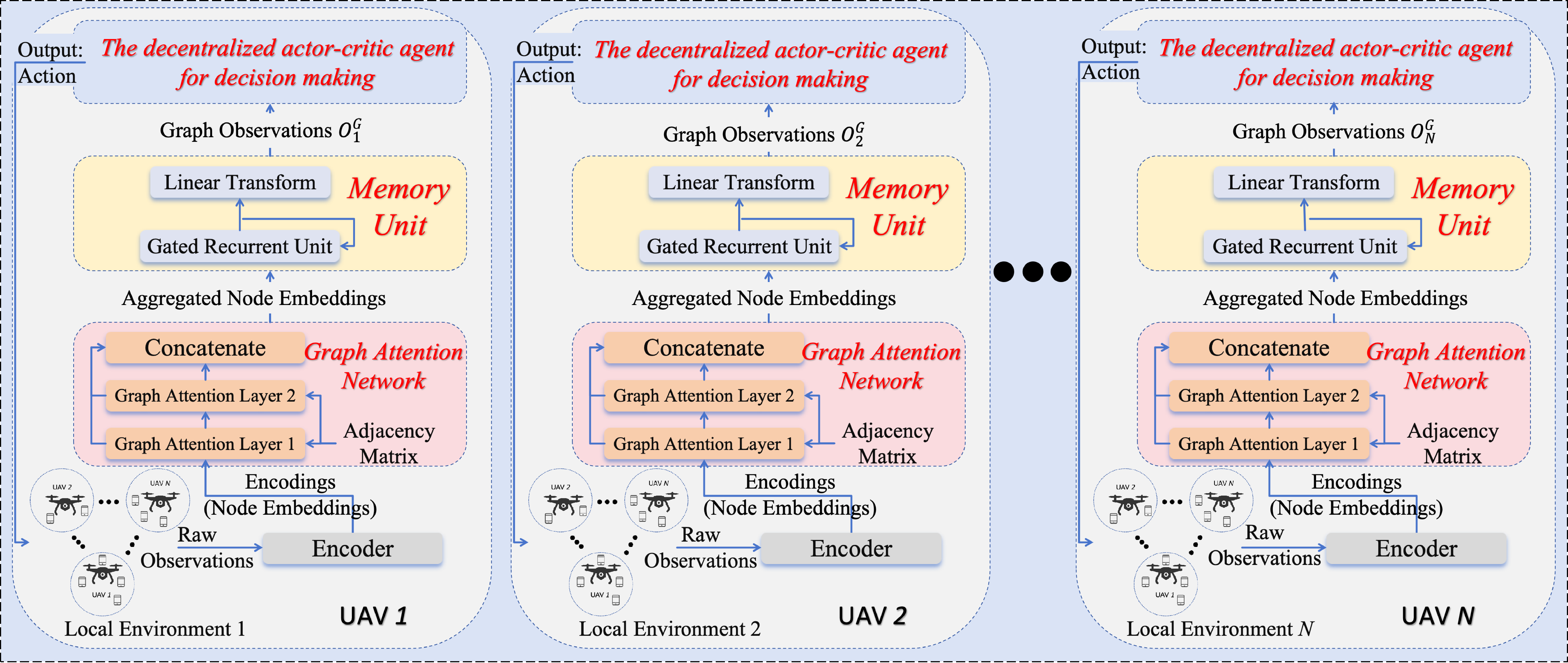}}
\caption{
Architecture of the proposed GADC framework.
}
\label{Architecture}
\end{figure*}

\begin{figure*}[!tbp]
\centering
\includegraphics[width=.95\textwidth]{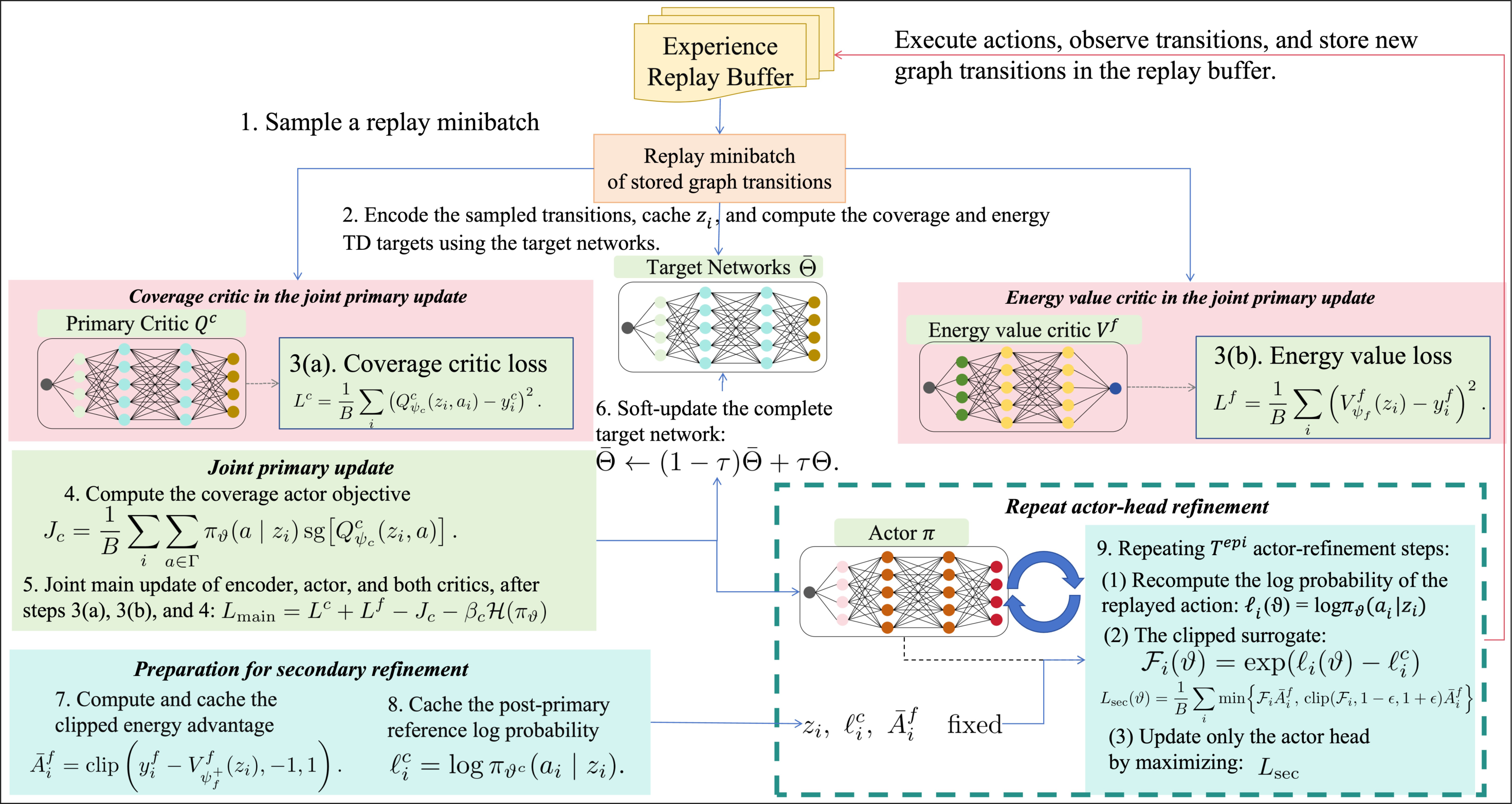}
\caption{
{One GADC training update. The secondary loop reuses fixed encoded observations, reference probabilities, and advantages, and updates only the actor head.}
}
\label{RL_architecture}
\end{figure*}

\section{GADC-Based Reinforcement Learning}
\label{Sec:DualAC}
The decentralized control of a multi-UAV system constitutes a sequential decision-making problem where each action directly influences both immediate service utility and future energy states. 
In this stochastic, dynamic environment, optimizing the long-term dual objectives naturally fits an MDP framework. This formulation provides a systematic approach to account for state transitions and the dependencies between decentralized agents over the entire operational horizon.

\subsection{Dual-Objective Markov Decision Process}
\label{Sec:problem}
{
We model the underlying global system dynamics as an MDP $(\mathcal{S},\Gamma^N,P,r,\gamma)$, where $\Gamma^N$ is the joint action space and $\Gamma$ is the individual UAV action set defined below. A constrained policy formulation specifies the primary-secondary objective ordering. At the beginning of time slot $t$, each UAV executes the shared decentralized policy using its partial local observation. The action changes the UAV position, service region, observation region, neighbor set, service assignment, and residual energy, which together determine the next state.
}

The decision variable is the shared decentralized policy $\pi$. The environment state includes UAV and UT positions, residual energies, service assignments, and channel information. Each UAV instead observes the local feature vector $s_n(t)$ specified in Eq.~\eqref{eq:local_state_vector}; this partial observation need not be Markov. Each UT is assumed to use a GPS-capable device to acquire and report its location~\cite{6295661}, which is available to observing UAVs without modeled feedback delay.

{
Each UAV selects an action from $\Gamma=\{\mathrm{Hover},N,S,W,E\}$.
A movement action specifies a horizontal waypoint increment of length $s_{\mathrm{mov}}$, whereas the hover action specifies position holding.
These commands define high-level navigation decisions for rotary-wing UAVs. 
In a physical deployment, a trajectory generator would convert successive commands into a continuous reference, which a low-level flight controller would track subject to velocity, acceleration, and jerk limits. 
Existing multicopter controllers provide such reference generation mechanisms~\cite{px4jerk}.
The joint action is $\bm{a}(t)=(a_1(t),\ldots,a_N(t))\in\Gamma^N$, and the transition kernel $P(s(t+1)\mid s(t),\bm{a}(t))$ describes the resulting changes in positions, service associations, connectivity, and energy.
}

{
During training, each UAV samples from the actor's softmax distribution. With probability $\zeta_e$, the sampled command is replaced by a uniformly drawn action. Thus,
\begin{equation}
\mu_e(a\mid O_n^G)=(1-\zeta_e)\pi_\vartheta(a\mid O_n^G)+\zeta_e/5,
\quad a\in\Gamma.
\label{eq:behavior_policy}
\end{equation}
Greedy evaluation selects the largest actor output. The exploration probability $\zeta_e$ is linearly annealed from $0.5$ to $0.02$ over training and is distinct from the secondary clipping parameter $\epsilon$.
}

{
The action-domain and service-association rules are incorporated directly into the environment dynamics. Battery depletion is handled separately through the terminal-state convention in Eqs.~\eqref{RESIDUALEE}--\eqref{eq:network_lifetime}.
Each UT is assigned to at most one UAV in a time slot, such that $\sum_{n=1}^{N}\omega_{m,n}^{s}(t)\leq 1$.
Each UAV selects a movement command from the discrete action set $\Gamma$. 
The residual battery state is updated according to~\eqref{RESIDUALEE}, where negative values are clipped to zero.
A transition for which any UAV reaches zero residual energy is terminal and defines the end of the network lifetime.
Let $\Pi_{\Gamma}$ denote the class of decentralized policies whose actions are supported on the discrete action set $\Gamma$.
}

The global coverage objective uses the instantaneous reward
\begin{equation}
    r^{c}(t)=\sum_{n=1}^{N}c_n(t),
\end{equation}
which measures the number of served UTs. The minimum-residual-energy surrogate is
\begin{equation}
    r^{f}(t)=\min_{n\in\{1,\ldots,N\}} b_n(t)/E_0,
\end{equation}
where $E_0$ is the nominal initial energy. This energy reward is a surrogate for delaying first depletion, not the network lifetime itself. Maximizing its discounted return is not equivalent to maximizing the stopping time. The simulator-specific training rewards are specified in Section~\ref{numerical result}.
For a shared policy $\pi$, the corresponding discounted returns are
\begin{equation}
\begin{aligned}
\mathcal{J}^{c}(\pi)
&=\mathbb{E}_{\tau\sim\pi}\left[\sum_{t=0}^{\infty}\gamma^{t}r^{c}(t)\right],\\
\mathcal{J}^{f}(\pi)
&=\mathbb{E}_{\tau\sim\pi}\left[\sum_{t=0}^{\infty}\gamma^{t}r^{f}(t)\right].
\end{aligned}
\end{equation}

{
Coverage is the primary mission objective, while network lifetime is the secondary operational objective. Thus, our goal is not to construct a full Pareto frontier over equally prioritized objectives, but to first establish a coverage-oriented policy and then refine it toward improved lifetime. Directly optimizing a weighted reward
$\varphi r^{c}(t) + (1-\varphi) r^{f}(t)$
does not explicitly preserve this primary-secondary structure and is sensitive to the numerical scales of the two reward terms. We therefore adopt a two-stage formulation.}

The idealized primary problem is
\begin{equation}
    \pi^{c}=\arg\max_{\pi\in\Pi_\Gamma}\mathcal{J}^{c}(\pi).
    \label{eq:primary_policy}
\end{equation}
The secondary reference problem improves the energy surrogate within a KL neighborhood of the coverage policy.
\begin{equation}
\begin{aligned}
\max_{\pi\in\Pi_\Gamma} \quad & \mathcal{J}^{f}(\pi) \\
\mathrm{s.t.} \quad &
\mathbb{E}_{s\sim d^{\pi^{c}}}
\left[
D_{KL}\left(\pi(\cdot|s)\,\|\,\pi^{c}(\cdot|s)\right)
\right]\le \delta ,
\end{aligned}
\label{cmdp}
\end{equation}
{
where $d^{\pi^{c}}$ is the discounted state visitation distribution and $\delta$ specifies the KL-neighborhood bound.
The service-association, action-domain, and battery-termination rules belong to the environment dynamics. The KL condition instead defines a policy neighborhood for the analysis.
The implementation uses a discrete-action coverage update followed by a replay-based clipped secondary update, rather than exactly solving either optimization problem.
Clipping neither enforces nor verifies the KL condition. Lemma~\ref{lemma1} therefore remains conditional on that condition.
}

\subsection{System Architecture}
{Ideally, global optimization over all UAV actions would jointly improve coverage and network lifetime.}
In practice, each UAV must make decentralized decisions based on its local observation and the messages received from neighboring UAVs.
To bridge this gap, we propose the GADC framework shown in Fig.~\ref{Architecture}.
The architecture contains four main modules.
First, each UAV maps its local observation into a node embedding through a shared encoder.
Second, adjacency-masked graph attention layers aggregate information from in-range neighboring UAVs.
Third, a GRU-based memory unit preserves temporal information across time slots.
Finally, the resulting graph-enhanced observation is fed into a shared actor and two critic heads.
{
The primary critic guides the coverage update, while the secondary critic supplies the energy-reward advantage for clipped actor refinement.
}
Fig.~\ref{RL_architecture} illustrates the training process of the actor-critic module.
{
The architecture couples one shared actor with two objective-specific critics.
A coverage-oriented actor update establishes the intermediate policy $\pi^c$, after which the energy-reward update refines the actor head using a fixed reference.
The GAT-GRU representations used in the secondary loop are cached before the primary update and detached from the computation graph.
Thus, the secondary loop changes the actor head without updating the encoder.
}

\begin{table}[!t]
\centering
\scriptsize
\renewcommand{\arraystretch}{1.15}
\caption{Ideal-simulator GADC network configuration for each UAV}
\label{tab:gadc_network}
\begin{tabular}{c|p{0.50\columnwidth}|c}
\hline \hline
\textbf{Block} & \textbf{Layer} & \textbf{Params} \\
\hline \hline
Input & Local observation $s_n(t)\in\mathbb{R}^{12}$ & $0$ \\
\hline
\multirow{5}{*}{Encoder}
& Embedding $12\!\to\!128$, $\tanh$ & $1{,}664$ \\
& Two GAT layers, $4$ heads with dimension $32$ & $99{,}072$ \\
& Fusion $256\!\to\!128$, ReLU & $32{,}896$ \\
& GRU cell $128\!\to\!128$ & $99{,}072$ \\
& Output projection $128\!\to\!128$, $\tanh$ & $16{,}512$ \\
\hline
Actor & $128\!\to\!5$, softmax & $645$ \\
Primary critic & $128\!\to\!5$, coverage $Q^c$ & $645$ \\
Secondary critic & $128\!\to\!1$, energy $V^f$ & $129$ \\
\hline
\multicolumn{2}{c|}{\textbf{Total}} & $\mathbf{250{,}635}$ \\
\hline
\end{tabular}
\end{table}
GADC is specialized for the decentralized coverage-lifetime control problem in three aspects.
First, all homogeneous UAVs share one graph-attention and recurrent encoder.
The same encoder is applied to every UAV node, so the learned policy can be used with different fleet sizes without changing the network architecture.
Second, the graph-attention layers are masked by the time-varying adjacency matrix.
Each UAV therefore aggregates information only from in-range neighboring UAVs, which matches the decentralized communication model.
Third, the shared graph-enhanced representation is connected to one actor head and two separate critic heads.
{
The primary critic estimates coverage action values, while the secondary critic estimates the discounted energy-reward return.
}
This structure differs from weighted-sum actor-critic methods, where the two objectives are collapsed into one scalar reward and evaluated by a single critic.

\begin{algorithm}[!tbp]
\caption{{GADC training}}
\label{algorithm:DGAC}
\DontPrintSemicolon

\KwIn{Fleet size $N$, horizon $T$, exploration schedule $\zeta_e$, clipping range $\epsilon$, secondary iterations $T^{epi}$.}
\KwOut{Shared actor and GAT-GRU encoder for decentralized execution.}
Initialize network $\Theta$, target network $\bar\Theta=\Theta$, and FIFO replay $\mathcal D$\;
\For{episode $e=1$ to MaxEpisodes}{
    Reset environment and recurrent states\;
    \For{$t=0$ to $T-1$}{
        Encode local observations, aggregate neighbors, and update recurrent states\;
        Sample each $a_n$ from $\mu_e$ in Eq.~\eqref{eq:behavior_policy}\;
        Execute joint action and record rewards, next observations, and terminal flag $d$\;
        Store graph observations, adjacencies, recurrent states, actions, rewards, and $d$ in $\mathcal D$\;
        \If{replay warmup and update schedule permit}{
            Uniformly sample graph transitions from $\mathcal D$\;
            Compute current representations $z_i$ and target next representations $z_i'$\;
            Compute and cache TD targets using Eq.~\eqref{TDerror}\;
            Update $\Theta$ with the joint primary loss in Eq.~\eqref{optimizeEffective}\;
            Soft-update $\bar\Theta$ by Eq.~\eqref{softupdates}\;
            Detach the cached $z_i$ and compute fixed advantages by Eq.~\eqref{eq:advantage}\;
            Cache $\ell_i^c=\log\pi_{\vartheta^c}(a_i\mid z_i)$ after the primary step\;
            \For{$k=1$ to $T^{epi}$}{
                Update only the actor head by maximizing Eq.~\eqref{eq:clipped_loss}\;
            }
        }
        \If{battery depletion or horizon reached}{End episode\;}
    }
}
\end{algorithm}

\subsection{Network Architecture}
\label{subsec:network_architecture}
This subsection specifies the input, hidden, and output layers of GADC.
Table~\ref{tab:gadc_network} gives the detailed GADC network configuration for each UAV.
\subsubsection{Input layer}\hspace{0.2em}
For UAV $n$, the local state observation $s_n(t)$ is represented by a $12$-dimensional feature vector
\begin{equation}
\begin{aligned}
s_n(t)=\big[
&\mathbf{p}_n(t), b_n(t), \ell_n(t),
\rho^{\mathrm{svc}}_n(t), \rho^{\mathrm{obs}}_n(t),\\
&\Delta^{\mathrm{svc}}_n(t),
\Delta^{\mathrm{obs}}_n(t),
\Delta^{\mathrm{near}}_n(t)
\big],
\end{aligned}
\label{eq:local_state_vector}
\end{equation}
where $\mathbf{p}_n$ is the normalized two-dimensional position, $b_n$ is the residual energy, $\ell_n$ is the service load, $\rho^{\mathrm{svc}}_n$ and $\rho^{\mathrm{obs}}_n$ are the fractions of UTs within the service and observation ranges, and $\Delta^{\mathrm{svc}}_n$, $\Delta^{\mathrm{obs}}_n$, and $\Delta^{\mathrm{near}}_n$ are the two-dimensional offsets to the serviceable-UT centroid, observable-UT centroid, and nearest UT, respectively.
\subsubsection{Hidden layers}\hspace{0.2em}
{
The encoder first maps $\mathbf{s}_n(t)\in\mathbb{R}^{12}$ to a
128-dimensional embedding using a linear layer with a tanh activation.
The embedding then passes through two graph-attention layers. Each
graph-attention layer uses $J=4$ attention heads with a head dimension
of 32. For each head, attention-weighted aggregation is performed over
the current set of in-range neighbors. Hence, the number of terms in
the aggregation varies with the network topology, whereas the output
dimension of each head remains 32. Concatenating the four heads produces
a 128-dimensional attended representation for each UAV.
The attended representation is concatenated with the 128-dimensional
input embedding and passed through a $256\rightarrow128$ fusion layer
with ReLU activation. A GRU cell then updates the 128-dimensional hidden
state of each UAV. Finally, a $128\rightarrow128$ projection with tanh
activation generates the graph-enhanced representation
$O_n^G(t)\in\mathbb{R}^{128}$.
}
\subsubsection{Output layers}\hspace{0.2em}
{The representation $O_n^G$ is shared by the actor, the primary critic, and the secondary critic.}
The actor outputs a probability distribution over the five movement commands $\{\mathrm{Hover},N,S,W,E\}$.
The primary critic outputs five coverage action values $Q^c(O_n^G,\cdot)\in\mathbb{R}^{5}$.
The secondary critic outputs an energy state value $V^f(O_n^G)\in\mathbb{R}$ for the clipped secondary refinement.

\subsection{Graph Attention and Recurrent Representation}
In this subsection, we detail the operation of the graph attention networks and memory units in Fig.~\ref{Architecture}. 
{Consider an undirected UAV graph $G(\mathcal{V},L)$, where $\mathcal{V}$ is the set of UAV nodes and $L$ is the set of direct UAV-UAV links. The time-varying adjacency matrix $\mathcal A(t)$ determines the one-hop neighbor set $\mathcal I_n(t)$ of UAV $n$.}
Each UAV's local raw observation contains the position, energy, load, and UT summaries defined in Eq.~\eqref{eq:local_state_vector}. Neighbor information enters through graph aggregation.
The local observation $s_{n}$ is transformed to an embedding $\mu_{n} = ENC(s_{n})$ through a common encoder function $ENC$.

{Fig.~\ref{GATExample} illustrates a neighborhood with heterogeneous service loads. The learned attention weights are computed from the complete node embeddings and depend jointly on service load, position, residual energy, and the other encoded local features. A larger service load alone does not imply a higher attention weight.}

\begin{figure}[!tbp]
\centerline{\includegraphics[width=.475\textwidth]{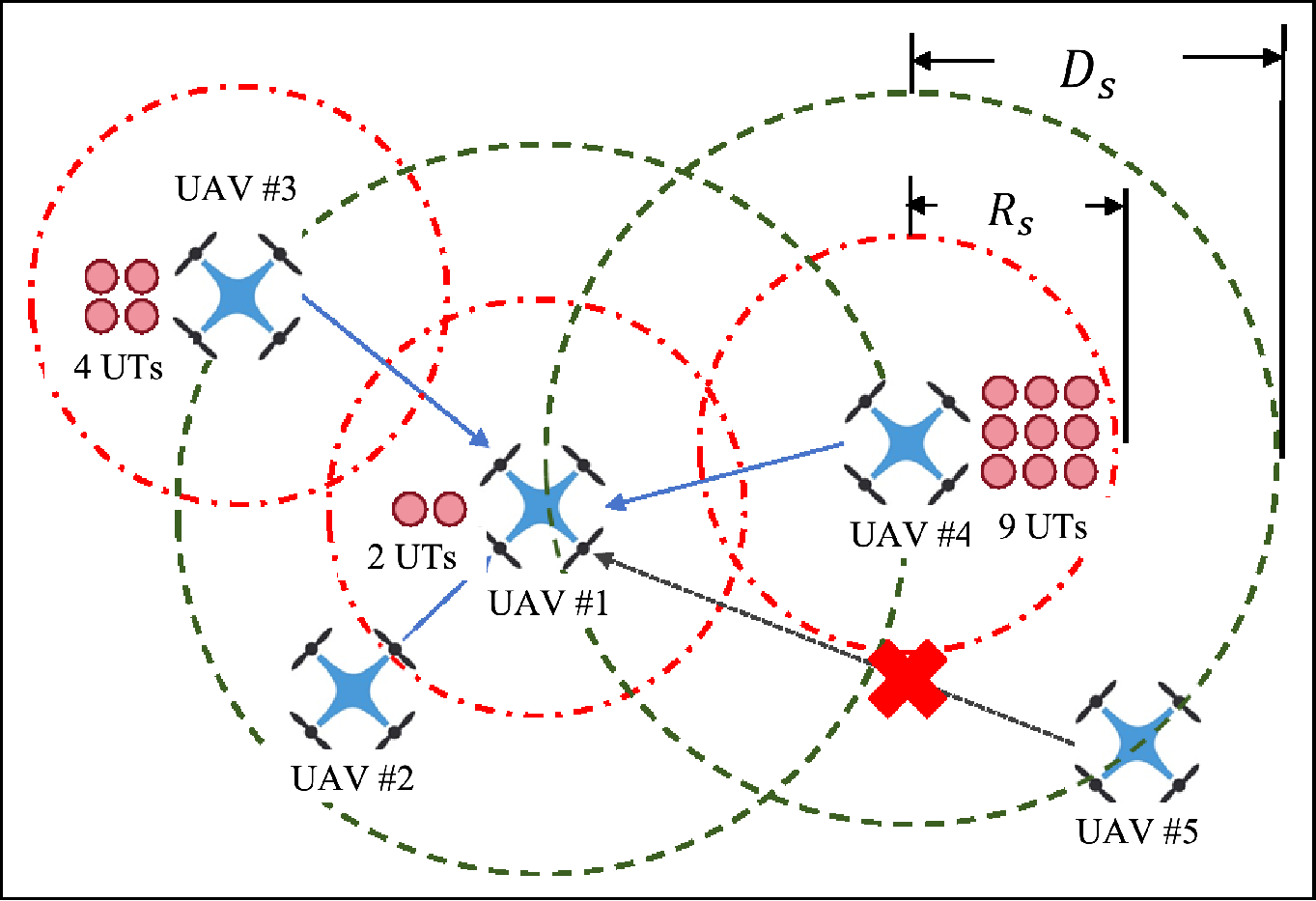}}
\caption{
Message passing between neighboring UAVs.
}
\label{GATExample}
\end{figure}

{For one attention head, let $\mu_n$ denote the input to the current layer and let $\mathcal N_n=\mathcal I_n\cup\{n\}$ include the self node. With head dimension $d_h=32$, adjacency-masked attention is}
\begin{equation}
    \begin{aligned}
        e_{n,i}&=\operatorname{ReLU}\left((W_K\mu_i)^T W_Q\mu_n/\sqrt{d_h}\right),\\
        \alpha_{n,i}&=\frac{\exp(e_{n,i})}{\sum_{l\in\mathcal N_n}\exp(e_{n,l})},
        \qquad i\in\mathcal N_n.
    \end{aligned}
    \label{attentionfun}
\end{equation}
where $\alpha_{n,i}$ is the attention weight that measures neighbor $i$'s importance to node $n$, computed using the self-attention mechanism's learnable matrices $W_{Q}$, $W_{K}$, and $W_{V}$ for Query, Key, and Value operations, respectively.
{The architecture employs multihead attention with $J$ independent heads to aggregate multiple learned feature projections.}
Each head $j$ has its own learnable matrices $W_{Q}^{j}$, $W_{K}^{j}$, and $W_{V}^{j}$.
Each node $n$ aggregates the embedding of neighboring nodes as 
\begin{equation}
\begin{aligned}
g_{n} = \text{Concat}_{j=1}^{J}(\sum\limits_{i\in \mathcal N_n}\alpha_{n,i}^{j}\cdot W_{V}^{j}\mu_{i}).
\end{aligned}
\end{equation}
{
With four 32-dimensional attention heads, the aggregated representation has a fixed dimension of 128, irrespective of $|\mathcal{I}_n|$.
The layer output is $\mu_n+\tanh(g_n)$, which becomes the input to the next GAT layer. After two layers, the output is fused with the original local embedding.
}

While stacking more GAT layers would aggregate information from more neighboring nodes and enhance embedding information, it also increases computational overhead. 
{
To balance the receptive field and computational cost, we use two GAT
layers. As described in Section~\ref{Sec: System Model}, the first
aggregation round exchanges local embeddings, while the second exchanges
the first-layer outputs. Information from a two-hop neighbor can therefore
influence UAV $n$ through an intermediate one-hop neighbor, without
requiring a direct physical link between them.
As shown in Fig.~\ref{Architecture}, the aggregated embeddings from the encoder and GAT layers are concatenated to produce the fused representation, which is then fed into the memory unit.
In a synchronous distributed implementation, the two layers require the two scheduled embedding-exchange rounds described in the communication model, while the graph and slot indices remain fixed within each control interval.
}

{
Let $\tilde g_n(t)$ denote the fused representation obtained by concatenating the original local embedding with the second-layer output and applying the linear fusion layer and ReLU activation.
The GRU updates its recurrent state at each slot according to
\begin{equation}
    \begin{aligned}
        h_{n}(t) = \text{GRU}\left(\tilde g_{n}(t), h_{n}(t-1)\right),
    \end{aligned}
    \label{memeryUnit}
\end{equation}
where $h_n(t-1)$ is the previous hidden state. A linear projection followed by tanh maps $h_n(t)$ to $O_n^G(t)$ at every slot. This representation is the input to the actor and both critics.
}

\subsection{Actor-Critic Networks}
\label{GADC_NETW}
{
The GAT-GRU output $z=O_n^G$ is the actor input under partial observability, not an assumed fully observed environment state.
We denote the intermediate policy after the primary update by $\pi^c$ and the policy after secondary refinement by $\pi^f$.
The two critics evaluate coverage and energy separately, and their actor updates are sequential rather than based on a weighted sum of rewards.

For the performance analysis, $s$ denotes the full network state and observation-action history, and $a$ denotes the joint action. Recurrent policies can then be represented on a common history-augmented state space.
For a common initial-state distribution, let
$d^\pi(s)=(1-\gamma)\sum_{t\geq0}\gamma^t P(s_t=s\mid\pi)$.
The global coverage return is
\begin{equation}
\mathcal J^c(\pi)=
\frac{1}{1-\gamma}\mathbb E_{s\sim d^\pi,a\sim\pi}[r^c(s,a)].
\label{eq:rigorous_objective}
\end{equation}
We write $\mathcal J^c_{\pi^c}$ and $\mathcal J^c_{\pi^f}$ for the returns before and after refinement. The corresponding analytical coverage advantage is $A^{\pi^c}(s,a)=Q^{\pi^c}(s,a)-V^{\pi^c}(s)$.
This full-state analysis is distinct from the sampled updates on encoded local observations.

Training uses a shared FIFO buffer of $2\times10^4$ graph transitions. A minibatch contains 64 transitions sampled uniformly without replacement, including their stored actions, adjacencies, recurrent states, rewards, and battery-terminal indicators.
Index $i$ enumerates the resulting $B$ UAV-transition samples.
The encoder computes $z_i$ using the current parameters and the stored recurrent state; the target encoder computes $z_i'$ for the next observation.
Let $d_i=1$ indicate battery depletion. A time-limit truncation alone does not set this flag.
Before the primary step, the target networks supply
\begin{equation}
\begin{aligned}
y_i^c&=r_i^c+\gamma(1-d_i)
\sum_{a\in\Gamma}\pi_{\bar\vartheta}(a\mid z_i')Q_{\bar\psi_c}^c(z_i',a),\\
y_i^f&=r_i^f+\gamma(1-d_i)V_{\bar\psi_f}^f(z_i').
\end{aligned}
\label{TDerror}
\end{equation}
Here $r_i^c$ and $r_i^f$ are the simulator-specific training rewards.
The critic losses are
\begin{equation}
\begin{aligned}
L^c&=\frac1B\sum_i\left(Q_{\psi_c}^c(z_i,a_i)-y_i^c\right)^2,\\
L^f&=\frac1B\sum_i\left(V_{\psi_f}^f(z_i)-y_i^f\right)^2.
\end{aligned}
\label{TDloss}
\end{equation}
The primary actor objective is
\begin{equation}
J_c=\frac1B\sum_i\sum_{a\in\Gamma}
\pi_\vartheta(a\mid z_i)\operatorname{sg}[Q_{\psi_c}^c(z_i,a)],
\label{eq:discrete_actor}
\end{equation}
where $\operatorname{sg}$ stops gradients through the critic output in this actor term.
Gradients are taken through the five action probabilities, not through a continuous action variable.
The encoder, actor, and both critics are updated jointly using Adam on
\begin{equation}
L_{\mathrm{main}}=L^c+L^f-J_c-\beta_c\mathcal H(\pi_\vartheta),
\label{optimizeEffective}
\end{equation}
where $\mathcal H$ is the minibatch-averaged action entropy and $\beta_c$ is its coefficient.
The complete target network then receives the soft update
\begin{equation}
\bar\Theta\leftarrow(1-\tau)\bar\Theta+\tau\Theta,
\qquad 0<\tau<1.
\label{softupdates}
\end{equation}
Here $\Theta$ includes the encoder, actor, and both critics.
This supplies slowly varying targets without implying a general reduction in TD-target bias or variance.

The secondary step reuses the replayed actions $a_i$. The representations $z_i$ computed before the primary step are cached and detached rather than recomputed with the updated encoder.
After the primary step, the energy advantage is evaluated as
\begin{equation}
\bar A_i^f=\operatorname{clip}
\left(y_i^f-V_{\psi_f^{+}}^f(z_i),-1,1\right),
\label{eq:advantage}
\end{equation}
where $\psi_f^{+}$ denotes the updated energy critic and $y_i^f$ is the target cached before the primary step.
The post-primary actor parameters $\vartheta^c$ provide the reference
$\ell_i^c=\log\pi_{\vartheta^c}(a_i\mid z_i)$.
For the current actor parameters $\vartheta$, define
\begin{equation}
\mathcal F_i(\vartheta)=
\exp\!\left(\log\pi_\vartheta(a_i\mid z_i)-\ell_i^c\right).
\label{eq:secondary_ratio}
\end{equation}
The replay-based clipped surrogate is
\begin{equation}
\begin{aligned}
L_{\mathrm{sec}}(\vartheta)=\frac1B\sum_i\min\big\{&
\mathcal F_i(\vartheta)\bar A_i^f,\\
&\operatorname{clip}(\mathcal F_i(\vartheta),1-\epsilon,1+\epsilon)
\bar A_i^f\big\}.
\end{aligned}
\label{eq:clipped_loss}
\end{equation}
Here $\epsilon>0$ is the likelihood-ratio clipping parameter.
The representations $z_i$, reference log probabilities $\ell_i^c$, and advantages $\bar A_i^f$ remain fixed throughout the $T^{epi}$ secondary steps. A separate Adam optimizer updates only the actor head to maximize $L_{\mathrm{sec}}$.

The reference policy generally differs from the behavior policy that generated the replayed action. The buffer does not store behavior log probabilities, and the ratio is not an importance-sampling correction for replay.
The update is a custom clipped off-policy refinement inspired by PPO~\cite{DBLP:journals/corr/SchulmanWDRK17}, not an on-policy PPO update.
Clipping removes the incentive for some further sampled-ratio changes but neither enforces nor verifies a prescribed KL bound.
It does not establish a policy-improvement guarantee for replayed samples. Its empirical effects are examined through ablation and sensitivity tests.
}

\subsection{Performance Bound Analysis}
{To characterize the coverage-return difference between a coverage-oriented policy and a candidate energy-refined policy, we use the \textbf{Performance Difference Lemma}~\cite{10.5555/645531.656005}.}
This lemma expresses the difference in expected performance between two policies in terms of the advantage function.
Formally, if we consider two policies \(\pi\) and \(\pi'\), the lemma is stated as
\begin{equation}
    \begin{aligned}
        \mathcal{J}(\pi')-\mathcal{J}(\pi)
        =
        \frac{1}{1-\gamma}
        \mathbb{E}_{\substack{s\sim d^{\pi'}\\a\sim\pi'}}
        \left[A^{\pi}(s,a)\right].
    \end{aligned}
\end{equation}
where \(\mathcal{J}(\pi)\) denotes the expected discounted return under policy \(\pi\), and \(A^{\pi}(s, a)\) is the advantage function, which measures how favorable it is to take action \(a\) in state \(s\) compared with the average action under policy \(\pi\).

{\( \mathcal{J}^{c}_{\pi^{f}} \) represents the coverage return achieved by the candidate energy-refined policy.}
{\( \mathcal{J}^{c}_{\pi^{c}} \) denotes the coverage return achieved by the baseline policy guided by the coverage action-value critic before secondary energy refinement.}

{
\begin{lemma}[Conditional Coverage-Change Bound under a KL Neighborhood]
\label{lemma1}
Let $\mathcal{J}^c(\pi)$ denote the expected discounted return for the
coverage objective. Given a baseline coverage policy $\pi^{c}$ and a
candidate energy-refined policy $\pi^{f}$ satisfying
\begin{equation}
\mathbb{E}_{s\sim d^{\pi^{c}}}
\left[
D_{KL}
\left(
\pi^{f}(\cdot|s)\,\|\,\pi^{c}(\cdot|s)
\right)
\right]
\leq \delta,
\label{eq:kl_condition}
\end{equation}
the coverage-return difference satisfies
\begin{equation}
\mathcal{J}^{c}_{\pi^{f}}-\mathcal{J}^{c}_{\pi^{c}}
\geq
\frac{1}{1-\gamma}
\mathop{\mathbb{E}}_{\substack{s\sim d^{\pi^{c}}\\a\sim\pi^{f}}}
\left[A^{\pi^{c}}(s,a)\right]
-
\frac{\sqrt{2\delta}\gamma\eta_{\pi^{f}}}
{(1-\gamma)^2},
\label{eq:final_bound}
\end{equation}
where $\delta$ specifies the KL-neighborhood bound and
$\eta_{\pi^{f}}
=
\max_s
\left|
\mathbb{E}_{a\sim\pi^{f}}
[A^{\pi^{c}}(s,a)]
\right|$.
\end{lemma}
Lemma~\ref{lemma1} provides a lower bound on the change in coverage
return when the policy is refined from $\pi^{c}$ to a candidate
$\pi^{f}$ satisfying the stated KL-neighborhood condition. The first
term in Eq.~\eqref{eq:final_bound} is the expected coverage advantage
under the baseline state distribution, while the second term accounts
for the policy-induced distribution-shift term and scales with
$\sqrt{\delta}$. The complete proof is provided in
Appendix~\ref{Proof}.
}

{
\begin{remark}
[Relation between Clipping and the KL Analysis]
The secondary update in Eq.~\eqref{eq:clipped_loss} clips ratios on replayed actions with fixed cached representations and a post-primary reference.
It neither enforces nor verifies the full-policy KL condition in Lemma~\ref{lemma1}.
The lemma applies only to candidate policies satisfying that condition and is not a policy-improvement guarantee for the implemented replay update.
\end{remark}
}

\subsection{Computational and Space Complexity Analysis}
{
The space complexity of the proposed GADC is dominated by the storage of neural network parameters and the replay buffer.
The complexity scales as
\begin{equation}
    \begin{aligned}
        \mathcal O\bigg(2(\nu_e+\nu_a+\nu_c+\nu_s)+|\mathcal D|[N(|a|+|s|+|h|)+N^2]\bigg),
    \end{aligned}
\end{equation}
where $\nu_a$, $\nu_c$, $\nu_s$, and $\nu_e$ are the network parameters corresponding to the actor, primary critic, secondary critic, and GAT-GRU encoder, respectively. The factor of two accounts for the target-network copy; replay stores recurrent states and dense adjacency matrices as well as observations and actions.
$|\mathcal{D}|$ is the number of stored graph transitions, bounded by the replay capacity of $2\times10^4$.
$|a|$ and $|s|$ are action and state vector dimensions, respectively.
The homogeneous UAV agents share the parameters of the GAT-GRU encoder, actor, and critic networks during training. For decentralized deployment, each UAV may maintain a local copy of the shared networks for inference. Therefore, the trainable network parameter count does not increase with \(N\), whereas the aggregate deployment memory scales linearly with \(N\) if a separate local copy is stored on each UAV.

For fixed numbers of GAT layers and attention heads, dense attention has a cost of $O(N|g|^2+N^2|g|)$ for projections and pairwise aggregation. A conservative per-UAV encoder bound is $C_{\rm enc}=C_{\rm loc}+N|g|^2+|h|^2$, where $C_{\rm loc}$ is the local observation-encoding cost, $|g|$ is the embedding dimension, and $|h|$ is the GRU hidden dimension. This bound also includes fusion and output projections at the fixed architecture dimensions.

Each slot gathers observations at a cost of $O(N(M+N))$ and evaluates the shared encoder and actor. During training, the primary step updates the encoder, actor, and both critics jointly, whereas the secondary steps update only the actor head. As in Section III-E, $B$ counts UAV-transition samples, so a minibatch of $b$ complete graph transitions contains $B=bN$ samples. Graph collation costs $O(B(|s|+|h|+N))$. The uniform sampling implementation may additionally incur $O(|\mathcal D|)$ work to select graph indices without replacement.
For an episode of $T$ slots with at most one replay update per slot, a conservative training-time bound is
\begin{equation}
\begin{aligned}
O\Big(T\big\{&N(M+N+C_{\rm enc}+\nu_a)+|\mathcal D|\\
&+B[C_{\rm enc}+\nu_c+\nu_s+(1+T^{\text{epi}})\nu_a]\big\}\Big).
\end{aligned}
\end{equation}
Here $T^{\text{epi}}$ is the number of secondary actor-head steps. The encoder term covers batched graph aggregation and recurrent processing; constant factors account for target-network evaluation and backpropagation.

Inference does not evaluate the critics or sample replay. The aggregate cost across all UAVs is bounded by
\begin{equation}
O\big(N[M+N+C_{\rm enc}+\nu_a]\big).
\end{equation}
For fixed feature dimensions and UT population, dense graph aggregation and neighbor discovery give quadratic worst-case fleet-size scaling. The $O(N^2|g|^2)$ term is a conservative upper bound rather than a tight expression when $|g|$ also varies.
The measured inference times on the tested workstation are reported in Table~\ref{tab:complexity_timing}.
}

\section{Performance Evaluation}
\label{numerical result}
We conduct a controlled ablation study of the proposed GADC controller and compare it with a high-budget sampled-search reference.
The ideal-simulator ablations use the same environment, training horizon, and random seeds. GAT-MADDPG uses one weighted-reward action-value critic without the secondary refinement. GAT-MAPPO uses an on-policy clipped update on the weighted reward instead of the coverage action-value update. MADDPG omits graph aggregation and uses local sensing features. Algorithm-specific optimizer and rollout settings are retained.
{The learned weighted-reward baselines combine per-UAV coverage and energy rewards through}
\begin{equation}
    \begin{aligned}
        {\hat{r}_n = \varphi r_n^{c} + (1-\varphi) r_n^{f}},
    \end{aligned}
    \label{eq:weighted_reward}
\end{equation}
{where $\varphi\in(0,1)$ is the weighting coefficient. The simulator-specific reward terms and default coefficients are specified below.}
{Sampled search instead evaluates each candidate joint action $\boldsymbol a$ in state $s$ using}
\begin{equation}
{
\begin{aligned}
S_{\mathrm{SS}}(s,\boldsymbol a)
&=\varphi_{\mathrm{SS}}\frac{C(s,\boldsymbol a)}{M}\\
&\quad +(1-\varphi_{\mathrm{SS}})\min_n\frac{b_n^{+}(s,\boldsymbol a)}{E_0}.
\end{aligned}}
\label{eq:sampled_search_score}
\end{equation}
{Here $C(s,\boldsymbol a)$ is the number of UTs served at the candidate UAV positions, and $b_n^{+}(s,\boldsymbol a)$ is the residual energy after deducting the candidate slot's modeled energy costs and clipping at zero. Under the homogeneous initialization used for sampled search, both score components lie in $[0,1]$. We use $\varphi_{\mathrm{SS}}=0.5$.}
Exact exhaustive enumeration requires evaluating $|\Gamma|^N$ joint actions at each time step, corresponding to a computational cost of $\mathcal{O}(|\Gamma|^N C_{\rm eval})$, where $C_{\rm eval}$ denotes the cost of evaluating one candidate joint action.
We therefore use a high-budget sampled-search reference that evaluates $K=10^7$ randomly sampled joint actions at each time step and selects the best sampled action. 
Its computational cost is $\mathcal{O}(K C_{\rm eval})$.
The hardware specifications and software versions used in our simulations are detailed in Table~\ref{tab:hardware}, with the simulation parameters provided in Table~\ref{tab:parameters}.
The Signal-to-Noise Ratio (SNR) is used as the service-link metric in the Sionna simulations.
We further evaluate GADC in a digital twin environment using the NVIDIA Sionna ray tracing engine. Unlike statistical channel models, Sionna accounts for the scene geometry and material properties of the urban map, including deterministic multipath effects and signal diffraction.

{
The reported normalized lifetime is
\begin{equation}
L_{\mathrm{norm}}=\frac{\min(T_{\mathrm{net}},T_{\max})}{T_{\max}},
\qquad T_{\max}=60.
\label{eq:normalized_lifetime}
\end{equation}
The depletion slot is included in the elapsed-slot count. Episodes without depletion within the horizon are assigned one, as are episodes depleted exactly at the final slot. This capped metric does not measure endurance beyond the horizon.
}

\begin{table}[!t]
    \centering
    \caption{Software \& hardware specifications}
    \begin{tabular}{c|l}
    \hline \hline
      \bfseries      \textbf{Device/Software} & \textbf{Specifications}   \\ \hline \hline
      CPU & Intel Xeon w9-3475X 2.20 GHz \\\hline
        GPU & NVIDIA RTX A6000 (48GB) \\\hline
              OS & Ubuntu 22.04.2\\\hline
      Python & 3.10.14  \\\hline
        PyTorch & 2.3.1 \\\hline
        TensorFlow & 2.15.1 \\\hline
        Sionna & 0.17.0 
        \\\hline
    \end{tabular}
    \label{tab:hardware}
    \vspace{-3pt}
\end{table}

\begin{table}[!t]
    \centering
    \caption{Sionna digital-twin simulation parameters}
    \begingroup
    \footnotesize
    \setlength{\tabcolsep}{3pt}
    \renewcommand{\arraystretch}{1.05}
    \begin{tabular}{>{\raggedright\arraybackslash}p{0.42\columnwidth}|>{\raggedright\arraybackslash}p{0.49\columnwidth}}
    \hline \hline
      \bfseries      \textbf{Parameters} & \textbf{Value}   \\ \hline \hline
      UT antenna height & 1.5 m\\\hline
      UAV altitude & 40 m  \\\hline
      Carrier frequency & 3.5 GHz\\\hline
      Bandwidth & 20 MHz\\\hline
      Transmit power & 30 dBm\\\hline
      Noise density & -174 dBm/Hz\\\hline
      Service-link metric & SNR\\\hline
      Service threshold & 6 dB for averaged held-out tests; 30 dB for cross-city visualization\\\hline
      UAV and UT setting & 10 UAVs, 60 UTs for averaged held-out tests; 5 UAVs, 60 UTs for cross-city visualization\\\hline
      UAV-UAV graph radius & 150 m for GADC graph message passing\\\hline
      Training ray tracing & $5\times10^4$ rays, FFT size 1, reflection depth 1, diffraction depth 0\\\hline
      High-fidelity evaluation & $10^6$ rays, FFT size 64, reflection depth 3, diffraction depth 3\\\hline
      Channel refresh & Every 5 slots in training; every slot in high-fidelity evaluation\\\hline
      Normalize delays & True\\\hline
    \end{tabular}
    \endgroup
    \label{tab:parameters}
    \vspace{-3pt}
\end{table}

\subsection{Performance evaluation in the Python/MATLAB simulator}
{
The ideal simulator uses a $200\,\mathrm{m}\times200\,\mathrm{m}$ area with service radius $R_s=24$ m, observation radius $R_o=55$ m, and UAV communication range $D_s=50$ m. The environment applies these radii directly to horizontal coordinates and does not simulate vertical motion or instantaneous small-scale fading. Under the geometric interpretation in Eq.~\eqref{eq:service_radius}, the nominal $R_s$ and $D_s$ imply an equivalent vertical separation of approximately $43.86$ m. This derived value is not an altitude input or a simulated state variable. The actor and critics use a hidden dimension of 128, and the attention mechanism uses four heads per GAT layer.
We train each learned method with 20 UAVs and 120 UTs for 5000 episodes using seeds 7, 17, and 27. Table~\ref{tab:ablation_holdout} evaluates these learned policies greedily on 50 held-out layouts with seeds 100000 to 100049. Fig.~\ref{fig:latest_scalability} evaluates the same checkpoints on five held-out layouts with seeds 100000 to 100004 at each fleet size, keeping 120 UTs and varying the number of UAVs from 20 to 40 without retraining. In both evaluations, each trained policy's metrics are first averaged over its test layouts, and the resulting policy-level averages are summarized across the three training seeds.
}

{
A UT is assigned to the nearest UAV with positive residual energy within its horizontal service radius.
Training uses local coverage rewards $r_n^c=Nc_n/M$, whose average is the global coverage fraction, and the common energy reward $r_n^f=\min_j b_j/E_0$.
The nominal initial energy is $E_0=35$, except in the stated sensitivity and heterogeneous-battery tests.
Unless otherwise stated, all learned weighted-reward baselines in the ideal simulator use $\varphi=0.5$ in Eq.~\eqref{eq:weighted_reward}, combining $r_n^c=Nc_n/M$ with $r_n^f=\min_j b_j/E_0$.
The primary and secondary Adam learning rates are $10^{-3}$ and $10^{-4}$, respectively, with three secondary steps, $\tau=0.01$, and $\beta_c=0$.
}

\begin{table}[!t]
    \centering
    \begingroup
    \caption{Held-out ablation results with 20 UAVs and 120 UTs}
    \setlength{\tabcolsep}{3.5pt}
    \begin{tabular}{c|c|c}
    \hline \hline
      \bfseries \textbf{Algorithms} & \textbf{Served UTs} & \shortstack{\textbf{Normalized}\\\textbf{lifetime}} \\ \hline \hline
      Sampled search  & 108.0 & 0.874 \\ \hline
      GAT-MADDPG & $88.6\pm1.2$ & $0.993\pm0.001$ \\ \hline
      \textbf{GADC} & $\mathbf{88.3\pm3.3}$ & $\mathbf{0.996\pm0.004}$ \\ \hline
      MADDPG & $78.0\pm0.2$ & $0.985\pm0.003$ \\ \hline
      GAT-MAPPO & $50.6\pm31.3$ & $0.986\pm0.012$ \\ \hline
    \end{tabular}
    \label{tab:ablation_holdout}
    \endgroup
    \vspace{-3pt}
\end{table}

\begin{figure*}[!t]
    \centering

    \includegraphics[
    width=.75\textwidth,
    viewport=5 205.42 507 226.69,
    clip
]{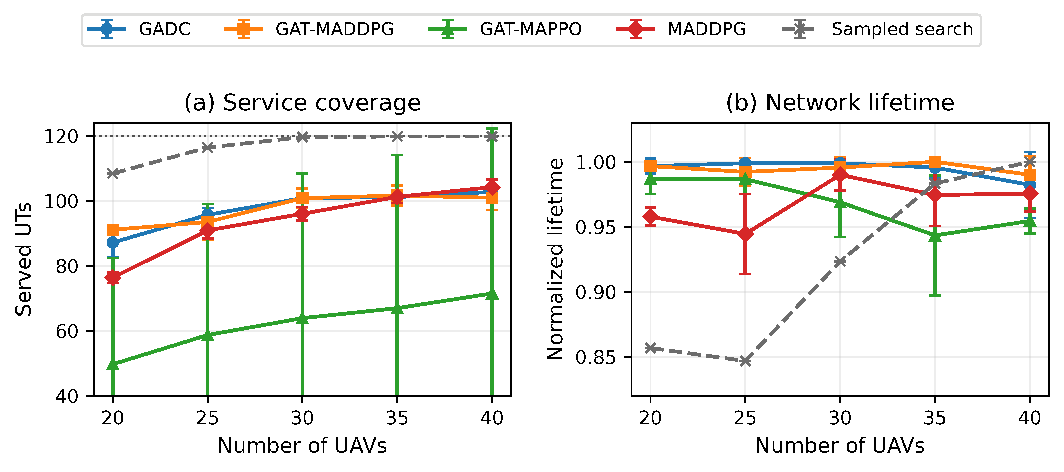}

    \subfigure[Service coverage.]{
        \includegraphics[
            width=.38\textwidth,
            viewport=4.94 9.16 246.84 169.44,
            clip
        ]{Figures/fig_latest_scalability_sampled_search.eps}
        \label{fig:latest_scalability_coverage}
    }
    \hspace{.04\textwidth}
    \subfigure[Normalized lifetime.]{
        \includegraphics[
            width=.38\textwidth,
            viewport=261.65 9.16 506.85 169.44,
            clip
        ]{Figures/fig_latest_scalability_sampled_search.eps}
        \label{fig:latest_scalability_lifetime}
    }

    \caption{Scalability results with varying UAV fleet sizes.}
    \label{fig:latest_scalability}
\end{figure*}

\begin{figure*}[!t]
    \centering
    \subfigure[Service coverage.]{
        \includegraphics[
            width=.38\textwidth,
            viewport=4 5 247 186.25,
            clip
        ]{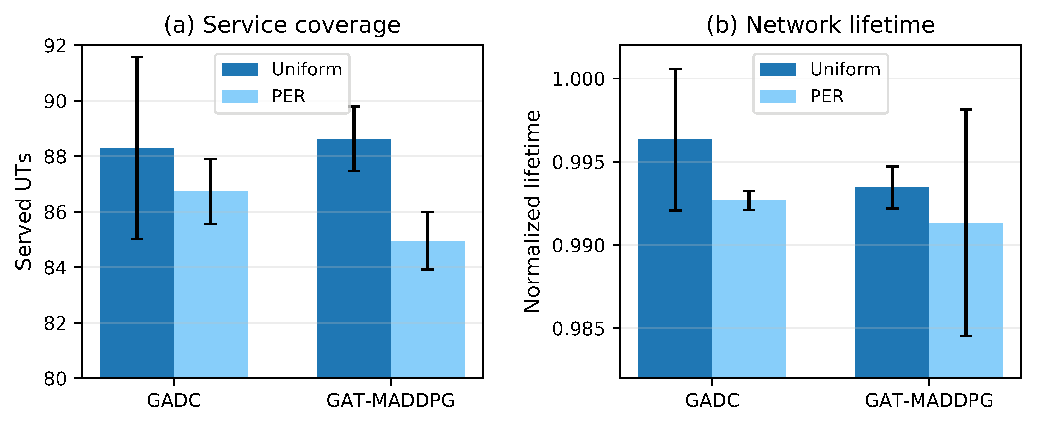}
        \label{fig:replay_coverage}
    }
    \hspace{.04\textwidth}
    \subfigure[Normalized lifetime.]{
        \includegraphics[
            width=.38\textwidth,
            viewport=250 5 497 186.25,
            clip
        ]{Figures/fig_latest_replay.eps}
        \label{fig:replay_lifetime}
    }
    \caption{Uniform replay and prioritized replay comparison on unseen layouts.}
    \label{fig:replay_comparison}
\end{figure*}

Table~\ref{tab:ablation_holdout} reports the held-out ablation results at the trained fleet size. 
The sampled-search reference serves 108.0 UTs and is included only as a high-budget non-learning reference. 
It evaluates \(K=10^7\) randomly sampled joint actions at each time step and is therefore impractical for real-time control.
GADC and GAT-MADDPG achieve close coverage, with $88.3\pm3.3$ and $88.6\pm1.2$ served UTs, respectively. 
GADC obtains a slightly higher normalized lifetime, $0.996\pm0.004$ compared with $0.993\pm0.001$. 
Removing the GAT encoder reduces coverage to $78.0\pm0.2$ served UTs, which shows the value of graph attention for multi-UAV coordination.
The on-policy GAT-MAPPO ablation without the coverage action-value update serves $50.6\pm31.3$ UTs, indicating lower mean coverage and greater seed variation in this setting.

Fig.~\ref{fig:latest_scalability} reports scalability when the 20 UAV policies are tested without retraining from 20 to 40 UAVs. In Fig.~\ref{fig:latest_scalability}(a), GADC increases from $87.1$ served UTs at 20 UAVs to $102.7$ served UTs at 40 UAVs. Fig.~\ref{fig:latest_scalability}(b) shows that GADC keeps normalized lifetime above $0.982$. GAT-MADDPG increases from $91.1$ to $100.9$ served UTs and keeps a lifetime of $0.990$ at 40 UAVs. MADDPG increases from $76.3$ to $104.2$ served UTs, but its lifetime at 40 UAVs is lower at $0.976$. {The mean coverage advantage of GADC and GAT-MADDPG over MADDPG is strongest at the trained scale and disappears at the highest tested fleet size, where MADDPG attains slightly higher coverage.} GAT-MAPPO remains the weakest coverage method, increasing only from $49.7$ to $71.5$ served UTs with large seed variation.

Fig.~\ref{fig:replay_comparison} compares uniform replay with prioritized experience replay (PER) on unseen layouts. In Fig.~\ref{fig:replay_comparison}(a), uniform replay achieves $88.3$ served UTs for GADC and $88.6$ served UTs for GAT-MADDPG. With PER, the corresponding values decrease to $86.7$ and $84.9$ served UTs. Fig.~\ref{fig:replay_comparison}(b) shows that the lifetime values remain close across samplers. {We therefore retain uniform replay. Under the tested settings, PER does not improve the held-out coverage of either GADC or GAT-MADDPG.}

\begin{table}[!t]
    \centering
    \caption{Held-out tradeoff performance comparison with 40 UAVs}
    \normalsize
    \setlength{\tabcolsep}{4pt}
    \begin{tabular}{c|c|c|c}
    \hline \hline
      \bfseries \textbf{Algorithms} & \shortstack{\textbf{Normalized}\\\textbf{lifetime}} & \textbf{Serv. UTs} & \shortstack{\textbf{Dominant}\\\textbf{scaling}} \\ \hline \hline
      MADDPG & 0.976 & 104.2 & $\mathcal{O}(N^2)$ \\ \hline
      GAT-MADDPG & 0.990 & 100.9 & $\mathcal{O}(N^2|g|^2)$ \\ \hline
      GAT-MAPPO & 0.954 & 71.5 & $\mathcal{O}(N^2|g|^2)$ \\ \hline
      Sampled search & 1.000 & 119.8 & $\mathcal{O}(K C_{\rm eval})$ \\ \hline
      \textbf{GADC} & \textbf{0.982} & \textbf{102.7} & $\mathcal{O}(N^2|g|^2)$ \\ \hline
    \end{tabular}
    \label{tab:tradeoff}
    \vspace{-3pt}
\end{table}

\begin{figure*}[!t]
    \centering
    \subfigure[Service coverage.]{
        \includegraphics[width=.38\textwidth,viewport=4.93 5.23 258.19 192,clip]{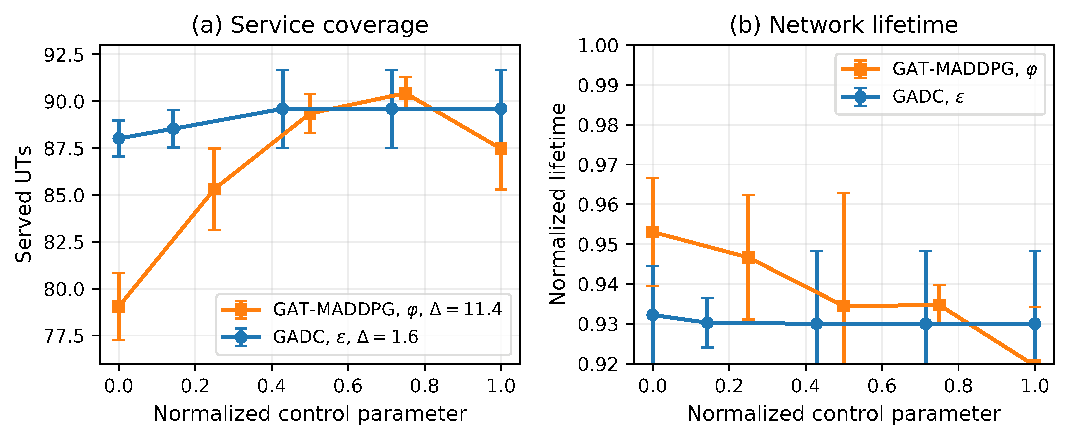}
        \label{fig:control_sensitivity_coverage}
    }%
    \hspace{.04\textwidth}
    \subfigure[Normalized lifetime.]{
        \includegraphics[width=.38\textwidth,viewport=259.84 5.23 511.45 193.125,clip]{Figures/fig_latest_control_sensitivity.eps}
        \label{fig:control_sensitivity_lifetime}
    }
    \caption{
    {Sensitivity to the weighted-sum coefficient $\varphi$ and the GADC likelihood-ratio clipping parameter $\epsilon$ with 20 UAVs and 120 UTs. Each parameter setting is trained independently.}}

    \label{fig:control_sensitivity}
\end{figure*}

\begin{figure*}[!t]
    \centering
    \includegraphics[height=9pt,viewport=10 45 76 56,clip]{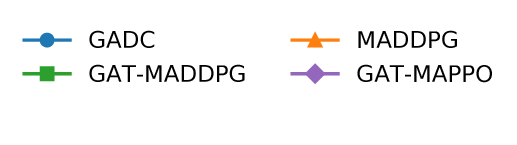}\hspace{1em}%
    \includegraphics[height=9pt,viewport=10 29 120 40,clip]{Figures/fig_r1_robustness_legend.eps}\hspace{1em}%
    \includegraphics[height=9pt,viewport=138 45 223 56,clip]{Figures/fig_r1_robustness_legend.eps}\hspace{1em}%
    \includegraphics[height=9pt,viewport=138 29 238 40,clip]{Figures/fig_r1_robustness_legend.eps}\\[1pt]
    \mbox{%
    \subfigure[Altitude proxy.]{
        \includegraphics[width=.36\textwidth]{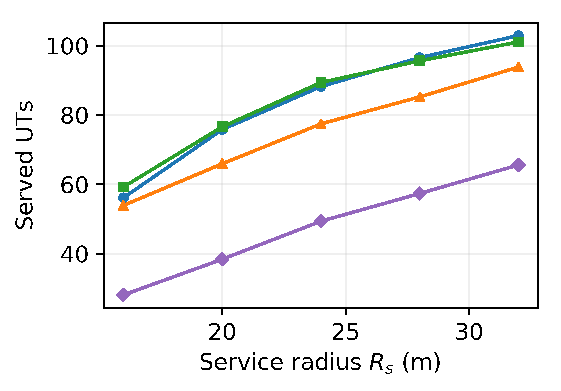}%
        \label{fig:r1_robustness_altitude}
    }\hspace{.04\textwidth}%
    \subfigure[Topology refresh interval.]{
        \includegraphics[width=.36\textwidth]{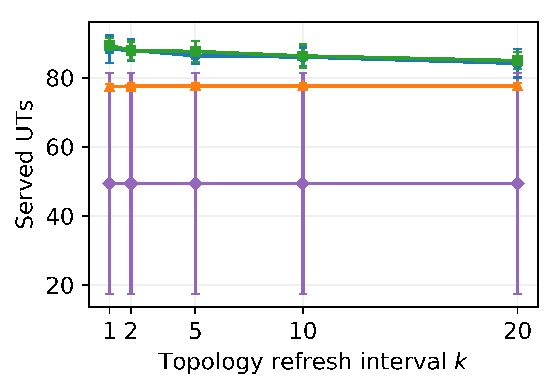}%
        \label{fig:r1_robustness_comm}
    }}\\[2pt]
    \mbox{%
    \subfigure[UT population.]{
        \includegraphics[width=.36\textwidth]{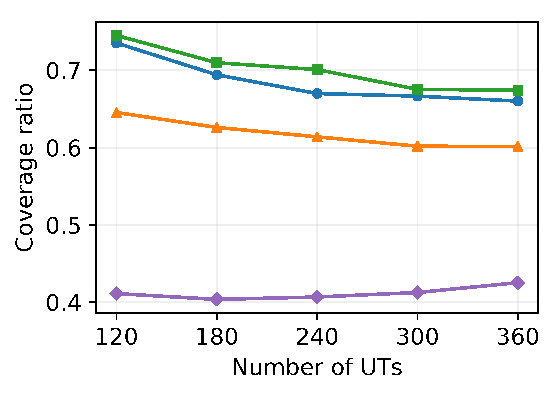}%
        \label{fig:r1_robustness_capacity}
    }\hspace{.04\textwidth}%
    \subfigure[State estimation noise.]{
        \includegraphics[width=.36\textwidth]{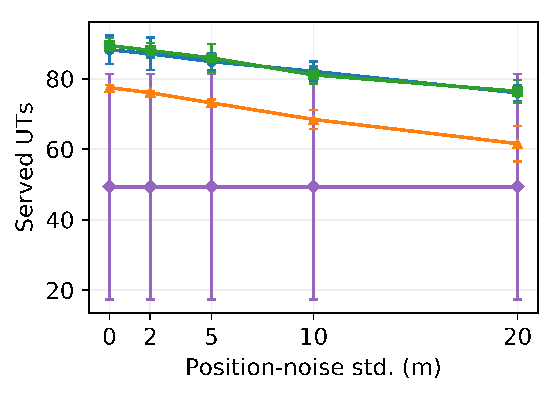}%
        \label{fig:r1_robustness_noise}
    }}
    \caption{
    Robustness tests for altitude proxy, topology refresh interval, UT population, and state estimation noise.}
    \label{fig:r1_robustness_suite}
\end{figure*}

\subsection{Service coverage vs. UAV lifetime tradeoff evaluation}

Table~\ref{tab:tradeoff} summarizes the 40 UAV held-out results and the corresponding computational complexity. {Its learned-policy entries use the same five-layout scalability evaluation as Fig.~\ref{fig:latest_scalability}.}
{The sampled-search reference provides the highest reference performance but remains infeasible for real-time control because its cost grows as $\mathcal{O}(K C_{\rm eval})$, where $K=10^7$. Table~\ref{tab:tradeoff} reports the dominant fleet-size scaling with fixed feature dimensions, hidden dimensions, and UT population.}
Among the learned methods, GADC gives a balanced result with $102.7$ served UTs and a normalized lifetime of $0.982$. MADDPG reaches $104.2$ served UTs but has a lower lifetime of $0.976$. GAT-MADDPG obtains a slightly higher lifetime of $0.990$ but lower coverage of $100.9$ served UTs. 
Together with Fig.~\ref{fig:control_sensitivity}, these results highlight the reduced sensitivity of GADC to tradeoff tuning across the tested operating points.

{
Fig.~\ref{fig:control_sensitivity} reports the nominal-density sensitivity study with 20 UAVs, 120 UTs, and an initial energy of 30.
Each candidate value is trained independently for 1500 episodes using seeds 7, 17, and 27 and is evaluated greedily on 20 held-out layouts with seeds 100000 to 100019. Sweeping $\epsilon$ from 0.05 to 0.4 changes the mean number of served UTs from 88.02 to 89.60, while the normalized lifetime remains between 0.930 and 0.932. 
Coverage reaches its highest mean value at $\epsilon=0.2$ and remains unchanged for the larger tested values. 
We therefore use $\epsilon=0.2$ as the smallest value on this observed plateau. Under the same protocol, sweeping $\varphi\in\{0.1,0.3,0.5,0.7,0.9\}$ for GAT-MADDPG produces a coverage range of 11.38 UTs, compared with 1.58 UTs for $\epsilon$. The horizontal axis in Fig.~\ref{fig:control_sensitivity} maps each parameter to $(x-x_{\min})/(x_{\max}-x_{\min})$ over its tested range.
}

\begin{table}[!t]
\centering
\caption{{Cross-density sensitivity to $\epsilon$ (mean $\pm$ standard deviation over three seeds).}}
\label{tab:epsilon_density}
\begingroup

\footnotesize
\renewcommand{\arraystretch}{1.08}
\setlength{\tabcolsep}{4pt}
\begin{tabular}{c|ccc}
\hline\hline
$N$ & $\epsilon=0.10$ & $\mathbf{\epsilon=0.20}$ & $\epsilon=0.30$ \\
\hline
\multicolumn{4}{c}{\textit{Served UTs}} \\
20 & $88.5\!\pm\!1.0$ & $\mathbf{89.6\!\pm\!2.1}$ & $89.6\!\pm\!2.1$ \\
30 & $104.2\!\pm\!0.3$ & $\mathbf{103.9\!\pm\!1.5}$ & $103.9\!\pm\!1.5$ \\
40 & $110.5\!\pm\!0.5$ & $\mathbf{110.5\!\pm\!1.1}$ & $110.5\!\pm\!1.1$ \\
\hline
\multicolumn{4}{c}{\textit{Normalized lifetime}} \\
20 & $0.930\!\pm\!0.006$ & $\mathbf{0.930\!\pm\!0.018}$ & $0.930\!\pm\!0.018$ \\
30 & $0.966\!\pm\!0.016$ & $\mathbf{0.975\!\pm\!0.009}$ & $0.975\!\pm\!0.009$ \\
40 & $0.967\!\pm\!0.023$ & $\mathbf{0.964\!\pm\!0.018}$ & $0.964\!\pm\!0.018$ \\
\hline\hline
\end{tabular}
\par\vspace{1pt}
{\scriptsize Bold denotes the default setting $\epsilon=0.20$.}
\endgroup
\end{table}

{
Table~\ref{tab:epsilon_density} examines $\epsilon$ at three spatial UAV densities. 
The map area and UT population remain fixed, while the fleet size increases from 20 to 30 and 40 UAVs. Unlike the policy transfer study in Fig.~\ref{fig:latest_scalability}, every fleet-size and clipping-parameter combination in this table is trained independently. 
With the default $\epsilon=0.2$, GADC serves $89.6\pm2.1$, $103.9\pm1.5$, and $110.5\pm1.1$ UTs at 20, 30, and 40 UAVs, respectively.
The corresponding normalized lifetimes are $0.930\pm0.018$, $0.975\pm0.009$, and $0.964\pm0.018$. 
Its largest difference from the best tested value at the same density is 0.35 served UT and 0.003 in normalized lifetime. 
Both differences are smaller than the corresponding standard deviation for the default setting. 
These results support the fixed default over the tested twofold density range, but they do not imply that the optimal clipping parameter is independent of every map size or deployment condition.
}

\subsection{Robustness and deployment tests}
\label{subsec:robustness}
We further examine whether the proposed controller remains stable when deployment conditions depart from the nominal setting.
For tests of service radius, topology refresh interval, UT population, and GPS noise, we evaluate the policies trained in the nominal setting without retraining.
 This isolates the sensitivity of the learned policy to each factor. For heterogeneous initial batteries, GAT-layer ablation, and dynamic UT mobility, the methods are retrained under the corresponding setting and then tested on held-out layouts.

Fig.~\ref{fig:r1_robustness_suite}(a) studies the altitude-dependent service footprint by sweeping the effective service radius $R_s$. The radius represents the net horizontal footprint induced by altitude, path loss, and LoS conditions. As $R_s$ increases from $16$ to $32$, GADC serves more UTs, from $56.1$ to $102.9$, while its normalized lifetime decreases from $1.000$ to $0.965$. This shows the expected coverage and energy tradeoff. At $R_s=32$, GADC keeps a longer lifetime than GAT-MADDPG and MADDPG, with $0.965$ compared with $0.951$ and $0.927$.

{
Fig.~\ref{fig:r1_robustness_comm} evaluates sensitivity to the topology refresh interval. 
The adjacency matrix is updated every $k$ slots and retained between updates, while current node features are used at each decision step. 
GADC serves 88.22, 86.00, and 84.16 UTs for $k=1$, $10$, and $20$, respectively, corresponding to reductions of 2.5\% and 4.6\% relative to $k=1$.
The same per-slot communication-energy rule is applied for all values of $k$. 
These results characterize sensitivity to delayed topology updates and do not quantify the energy savings of intermittent broadcasting.
}

{Fig.~\ref{fig:r1_robustness_suite}(c) examines increasing UT population with 20 UAVs. As the population increases from $120$ to $360$, GADC serves $88.2$ to $237.6$ UTs, while the coverage fraction decreases from $0.735$ to $0.660$. GAT-MADDPG gives similar coverage, whereas MADDPG serves fewer UTs. These results describe load sensitivity without an explicit per-UAV admission-capacity limit.}

{Fig.~\ref{fig:r1_robustness_suite}(d) injects Gaussian noise into observed positions while keeping the true dynamics unchanged. At $\sigma=5$ m, approximately $20\%$ of the service radius, GADC loses $3.8\%$ served UTs. At $\sigma=20$ m, GADC loses $13.9\%$, while MADDPG loses $20.5\%$. GADC therefore shows a smaller coverage reduction than the non-graph MADDPG baseline in this comparison.}

Table~\ref{tab:hetero_battery} reports the heterogeneous-battery test, where $E_0^n$ is sampled uniformly from $[0.6E_0,1.4E_0]$ for each UAV. GADC remains stable across seeds and serves $85.9\pm0.6$ UTs, which is only $2.7\%$ below the homogeneous-battery case. Normalized lifetime is lower than in the homogeneous setting for all methods. First depletion depends on both the initial energy distribution and the expenditure due to movement, service load, and communication.

\begin{table}[!t]
    \centering
    \begingroup
    \normalsize
    \renewcommand{\arraystretch}{1}
    \caption{Held-out results with heterogeneous initial batteries}
    \setlength{\tabcolsep}{6pt}
    \begin{tabular}{c|c|c}
    \hline \hline
      \bfseries \textbf{Algorithms} & \textbf{Served UTs} & \shortstack{\textbf{Normalized}\\\textbf{lifetime}} \\ \hline \hline
      GAT-MADDPG & $87.3\pm1.9$ & $0.832\pm0.011$ \\ \hline
      GADC & $85.9\pm0.6$ & $0.802\pm0.009$ \\ \hline
      MADDPG & $78.0\pm0.8$ & $0.897\pm0.007$ \\ \hline
      GAT-MAPPO & $68.6\pm0.0$ & $0.874\pm0.000$ \\ \hline
    \end{tabular}
    \label{tab:hetero_battery}
    \endgroup
    \vspace{-3pt}
\end{table}

{Table~\ref{tab:gat_layer_ablation} and Fig.~\ref{fig:r1_gat_layers} compare one, two, and three GAT layers. One layer provides one-hop aggregation and serves $85.6\pm2.8$ UTs. Two layers provide the highest mean coverage at $88.3\pm3.3$ served UTs with moderate update cost. A third layer increases the parameter count and update time but does not improve mean coverage. We therefore use two layers to balance coverage and computational cost.}

\begin{table}[!t]
    \centering
    \begingroup
    \normalsize
    \renewcommand{\arraystretch}{1}
    \caption{GAT layer ablation and update cost}
    \setlength{\tabcolsep}{6pt}
    \begin{tabular}{c|c|c|c}
    \hline \hline
      \bfseries \textbf{Layers} & \textbf{Served UTs} & \textbf{Params} & \textbf{Update ms} \\ \hline \hline
      1 & $85.6\pm2.8$ & $201{,}099$ & $19.65$ \\ \hline
      \textbf{2} & $\mathbf{88.3\pm3.3}$ & $\mathbf{250{,}635}$ & $\mathbf{24.10}$ \\ \hline
      3 & $86.0\pm1.8$ & $300{,}171$ & $31.31$ \\ \hline
    \end{tabular}
    \label{tab:gat_layer_ablation}
    \endgroup
\end{table}

\begin{figure}[!t]
    \centering
    \includegraphics[width=.38\textwidth]{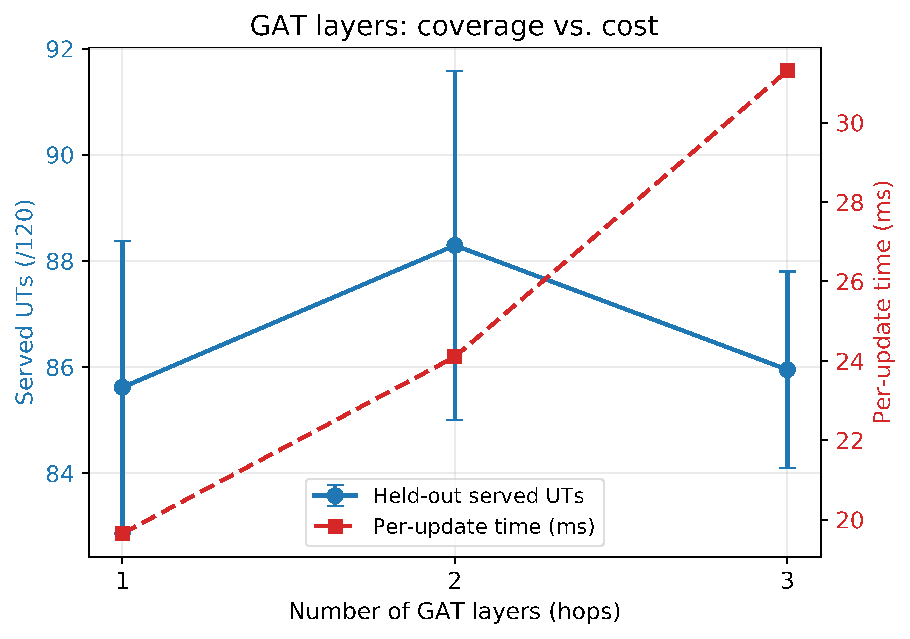}
    \caption{
    GAT depth ablation.}
    \label{fig:r1_gat_layers}
\end{figure}

\begin{figure}[!t]
    \centering
    \includegraphics[width=.38\textwidth]{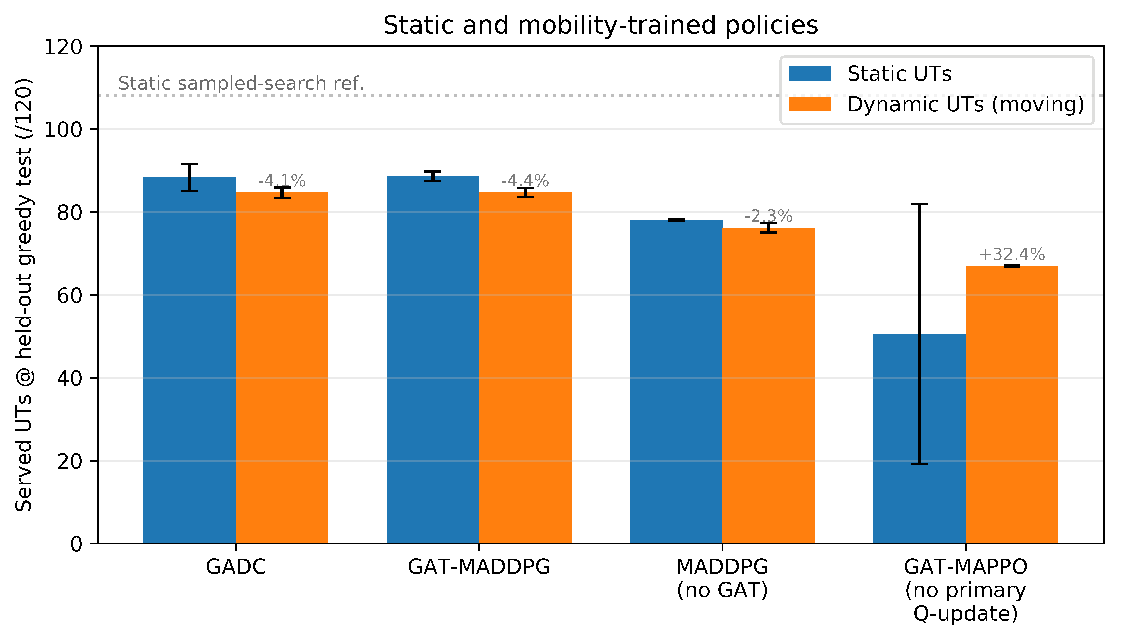}
    \caption{
   Dynamic UT mobility test.}
    \label{fig:r1_dynamic_ut}
\end{figure}

{Fig.~\ref{fig:r1_dynamic_ut} compares policies trained separately under static and dynamic UT conditions. Each mobile UT performs a bounded random walk with step size up to $5$ m every five slots and remains within $15$ m of its initial position. GADC serves $84.7\pm1.3$ UTs under this mobility, a $4.1\%$ reduction from the static case, while the normalized lifetime remains $0.996\pm0.002$. The ranking among the action-value-based methods is preserved, showing that the retrained policies can operate under the tested slow-mobility process.}

\begin{table}[!t]
    \centering
    \begingroup
    \normalsize
    \renewcommand{\arraystretch}{1}
    \caption{Measured GADC per-step and per-update time}
    \setlength{\tabcolsep}{6pt}
    \begin{tabular}{c|c|c}
    \hline \hline
      \bfseries \textbf{UAVs} & \textbf{Act ms} & \textbf{Update ms} \\ \hline \hline
      10 & $0.49$ & $32.44$ \\ \hline
      20 & $0.51$ & $29.62$ \\ \hline
      30 & $0.62$ & $42.82$ \\ \hline
      40 & $0.66$ & $60.09$ \\ \hline
      50 & $1.50$ & $114.92$ \\ \hline
    \end{tabular}
    \label{tab:complexity_timing}
    \endgroup
    \vspace{-3pt}
\end{table}
\FloatBarrier

Table~\ref{tab:complexity_timing} reports the measured act and update time as the fleet size grows. The forward action step remains below $1.6$ ms up to 50 UAVs. The update time grows from $32.44$ ms at 10 UAVs to $114.92$ ms at 50 UAVs, about $3.5\times$ for a $5\times$ larger fleet. These measurements characterize computational cost rather than convergence. Unlike CPO, GADC does not require Fisher-vector products or conjugate-gradient solves during training.

\begin{figure*}[!t]
    \centering
    \subfigure[Munich 3D map.]{
        \includegraphics[width=.27\textwidth]{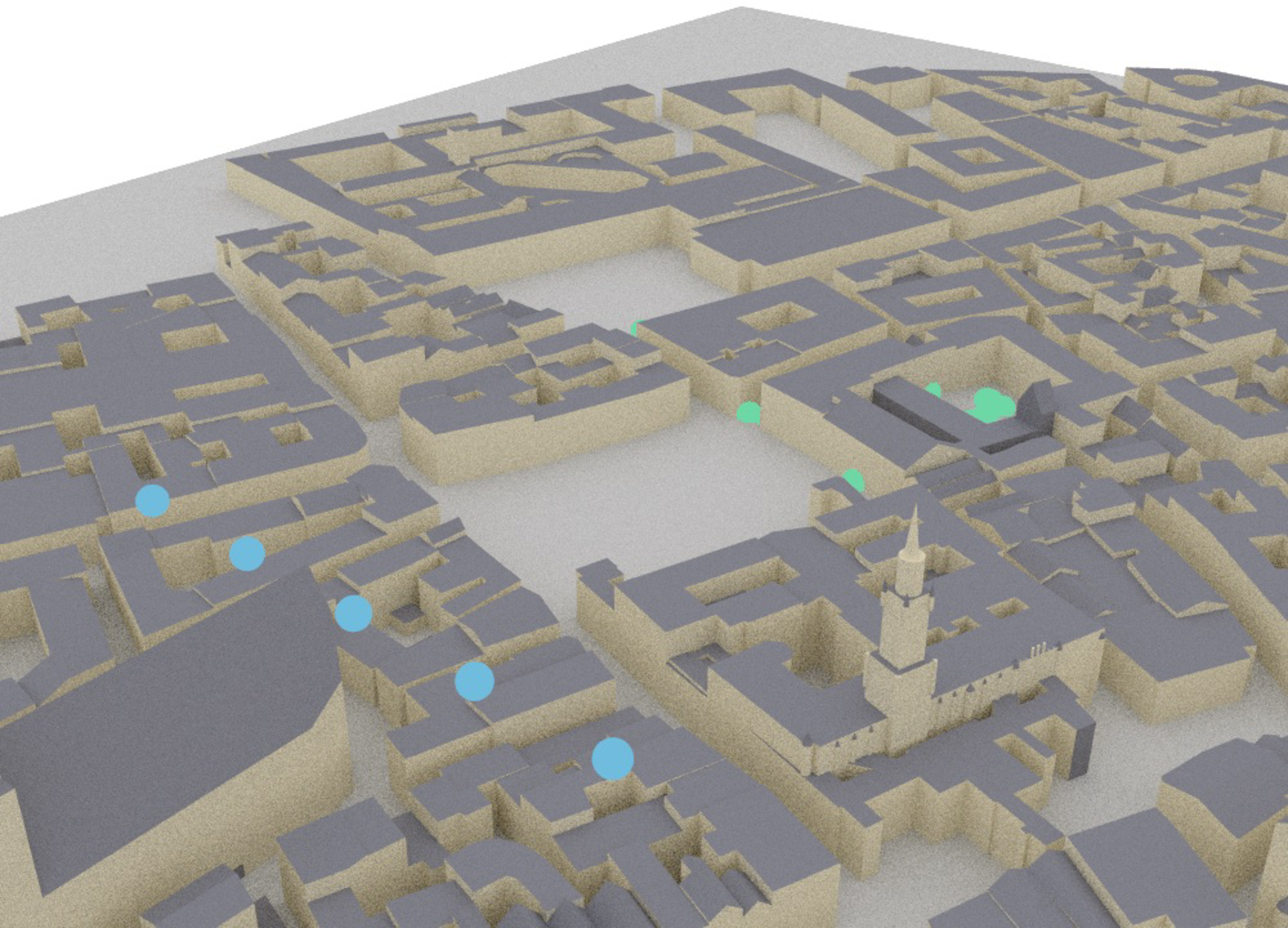}
        \label{fig:cross_city_munich_initial}
    }
    \subfigure[Munich GADC, 33/60 UTs.]{
        \includegraphics[width=.27\textwidth]{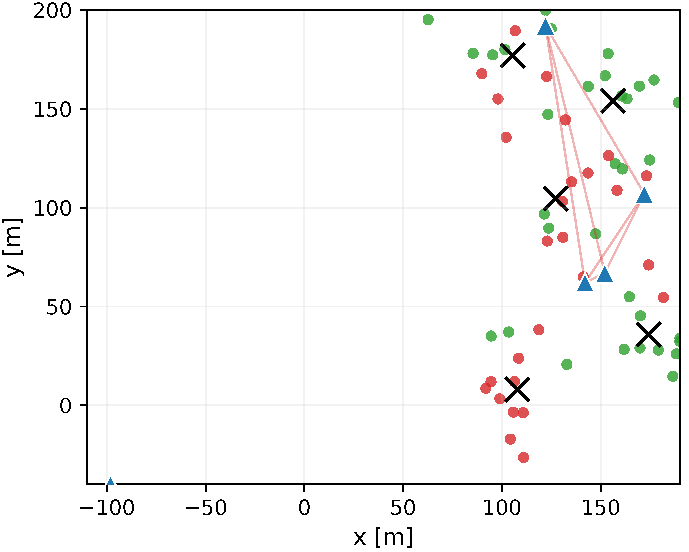}
        \label{fig:cross_city_munich_gadc}
    }
    \subfigure[Munich MADDPG, 9/60 UTs.]{
        \includegraphics[width=.27\textwidth]{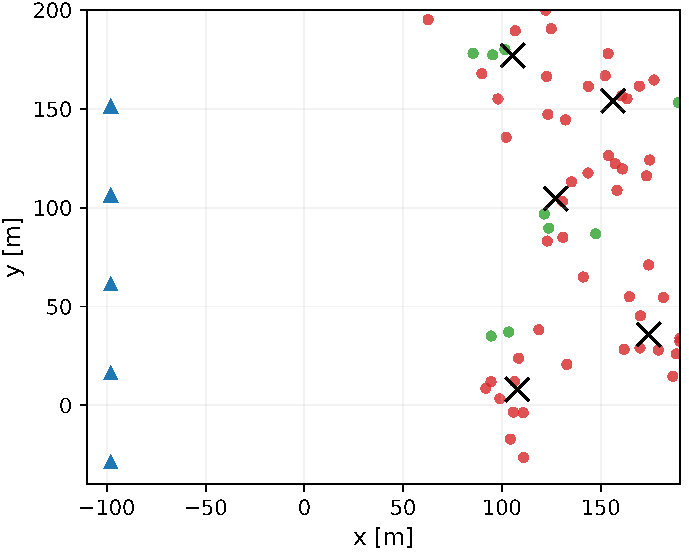}
        \label{fig:cross_city_munich_maddpg}
    }
    \subfigure[Etoile 3D map.]{
        \includegraphics[width=.27\textwidth]{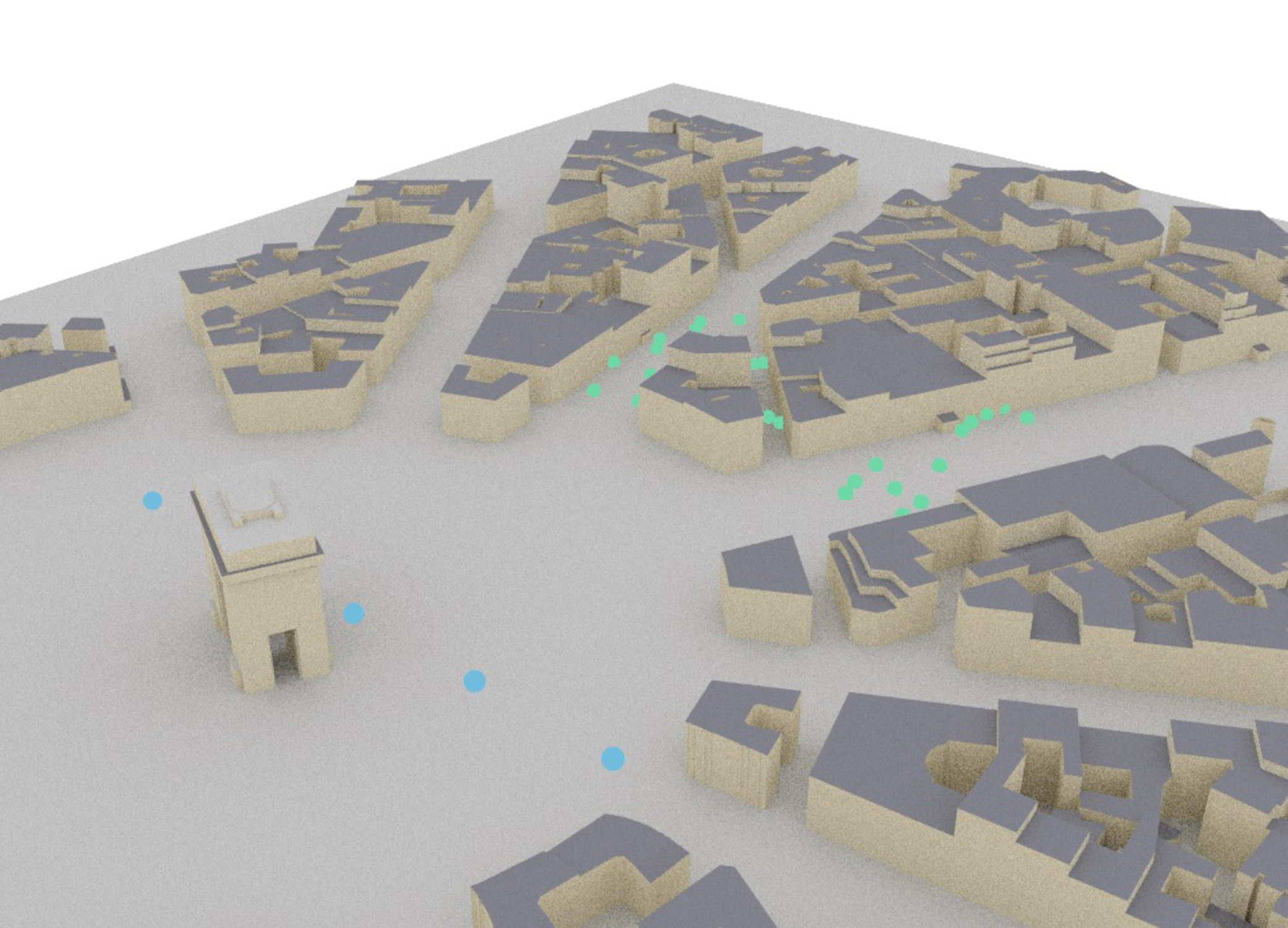}
        \label{fig:cross_city_etoile_initial}
    }
    \subfigure[Etoile GADC, 38/60 UTs.]{
        \includegraphics[width=.27\textwidth]{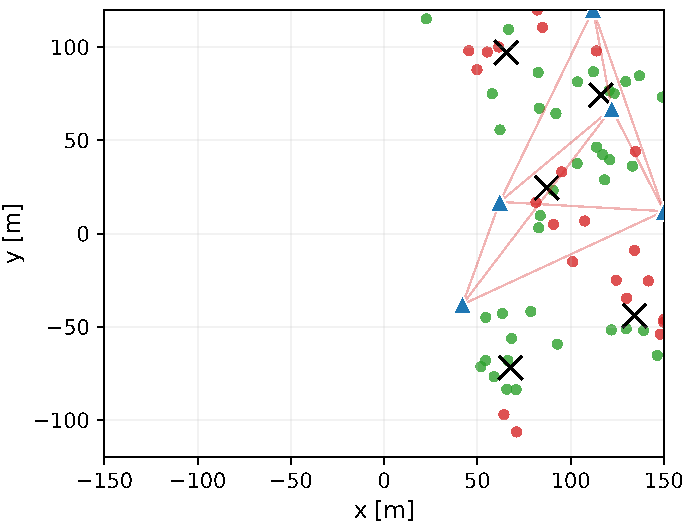}
        \label{fig:cross_city_etoile_gadc}
    }
    \subfigure[Etoile MADDPG, 27/60 UTs.]{
        \includegraphics[width=.27\textwidth]{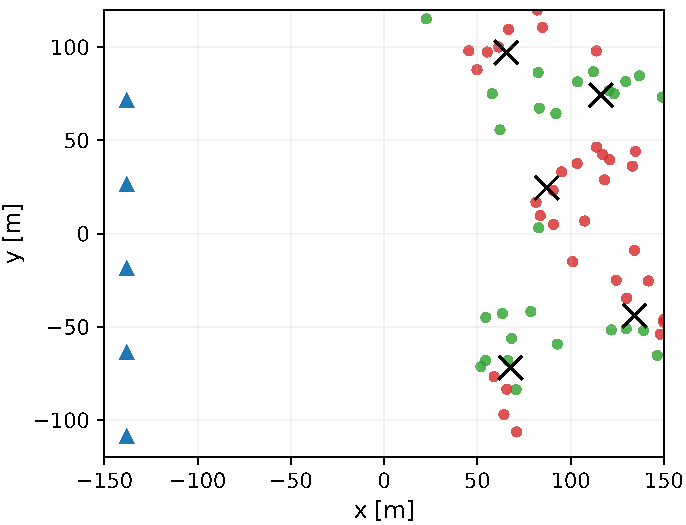}
        \label{fig:cross_city_etoile_maddpg}
    }
    \includegraphics[width=.475\textwidth]{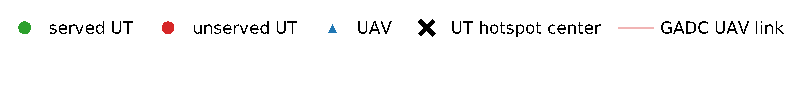}
    \caption{
Cross-city digital-twin deployments on Munich and Etoile. Green and red dots denote served and unserved UTs, blue triangles denote UAVs, black crosses denote hotspot centers, and red lines denote GADC graph links.}
    \label{fig:cross_city_deployment}
\end{figure*}

\begin{figure}[!t]
    \centering
    \subfigure[Service coverage.]{
        \includegraphics[width=.36\textwidth]{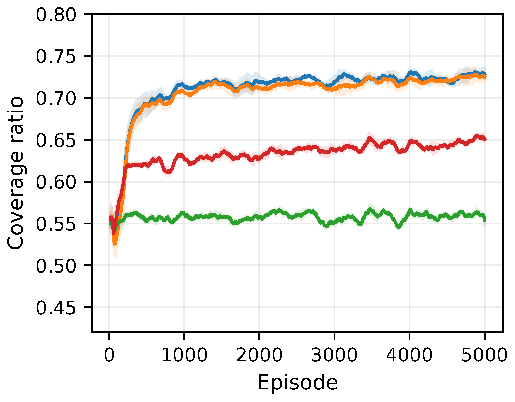}
        \label{fig:ideal_training_coverage}
    }%
    \hfill
    \subfigure[Normalized lifetime.]{
        \includegraphics[width=.36\textwidth]{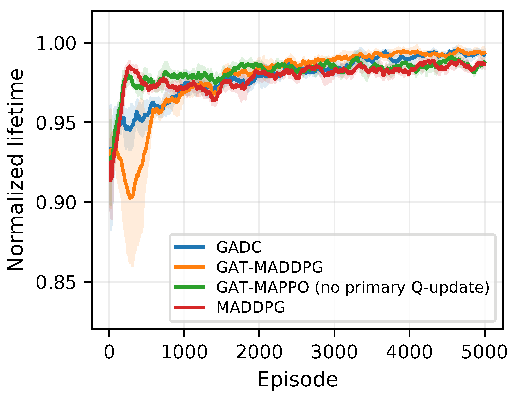}
        \label{fig:ideal_training_lifetime}
    }
    \caption{
    Learning curves for service coverage and network lifetime in the ideal simulator. 
    }
    \label{fig:ideal_training_curves}
\end{figure}

\subsection{Performance evaluation in the digital-twin simulator}
{The digital-twin simulator uses site-specific 3D urban maps. The averaged held-out evaluation is conducted in Munich, while the cross-city visualization additionally considers Etoile.}
We leverage \textbf{NVIDIA Sionna}~\cite{sionna2022} to perform accelerated ray tracing, which captures deterministic propagation characteristics, including LoS, specular reflections, and diffraction, derived from the exact scene geometry and electromagnetic material properties (per ITU-R P.2040)~\cite{series2025effects}. 
A valid service link is established only if the SNR exceeds the threshold in Table~\ref{tab:parameters}.
Due to the substantial computational cost of real-time ray tracing, training uses a lighter ray-tracing configuration with fewer launched rays, a single FFT sample, and slower channel refresh.
Held-out testing and cross-city visualization then use the high-fidelity ray-tracing configuration in Table~\ref{tab:parameters}.

{
When the cached channel state is refreshed, each UT selects an eligible UAV with positive residual energy by maximizing
$\mathrm{SNR}_{m,n}^{\mathrm{dB}}+0.25b_n/E_0$.
The second term is a numerical association preference, not a change in physical SNR.
The ray-tracing policies use local coverage rewards $r_n^c=Nc_n/M$ and the energy reward
\begin{equation}
r_n^{f,\mathrm{RT}}=0.65\,\min_j b_j/E_0+0.35\,b_n/E_0.
\label{eq:rt_energy_reward}
\end{equation}
This mixed reward differs from the common minimum-energy reward in the ideal simulator and is used for both compared methods.
MADDPG combines these terms as $\hat r_n=0.5r_n^c+0.5r_n^{f,\mathrm{RT}}$, whereas GADC uses separate coverage and energy critics without reward scalarization.
The 10-UAV model configurations assign an effective per-link energy coefficient of $0.02$ to both methods. In the separate 5-UAV visualization, this coefficient is zero for both methods.
These settings do not isolate the physical energy cost of GAT messaging.
}

In the held-out Sionna evaluation with 10 UAVs and 60 UTs, GADC serves $43.3\pm3.2$ UTs and obtains a normalized lifetime of $0.911\pm0.038$. MADDPG serves $41.7\pm4.6$ UTs and obtains a normalized lifetime of $0.950\pm0.050$. These averages are computed from three Munich test seeds with the 6 dB SNR threshold in Table~\ref{tab:parameters}. They are separate from the single-seed cross-city visualization in Fig.~\ref{fig:cross_city_deployment}, which uses 5 UAVs, 60 UTs, and the 30 dB visualization threshold. The averaged results show a coverage-lifetime tradeoff under the evaluated urban channels. GADC attains higher mean coverage, whereas MADDPG attains a higher normalized lifetime. The dense urban geometry and finite energy budget constrain coverage and episode duration in this evaluation. Because the saved 10-UAV configurations assign the same effective link-energy term to both methods, the lifetime difference does not isolate the physical energy cost of GAT messaging.

The ray-tracing experiment uses $\epsilon=0.2$ from the controlled sensitivity study and does not retune it for the ray tracing environment. Fig.~\ref{fig:cross_city_deployment} visualizes deployment behavior on Munich and Etoile with the same default $\epsilon$. Fig.~\ref{fig:cross_city_deployment}(a) and Fig.~\ref{fig:cross_city_deployment}(d) show the 3D city geometry with the initial UAV and UT locations. Fig.~\ref{fig:cross_city_deployment}(b) and Fig.~\ref{fig:cross_city_deployment}(e) show the final GADC deployments, while Fig.~\ref{fig:cross_city_deployment}(c) and Fig.~\ref{fig:cross_city_deployment}(f) show the final MADDPG deployments. In this single-seed visualization with 5 UAVs and 60 UTs, GADC serves 33 UTs in Munich and 38 UTs in Etoile, while MADDPG serves 9 and 27 UTs, respectively. From the same initial UAV locations, GADC moves the UAVs closer to the UT clusters, whereas MADDPG shows more conservative relocation. Graph aggregation permits neighboring observations to influence GADC decisions, but the policies also differ in their observation features and update rules. The visualization therefore provides a qualitative comparison rather than isolating the causal contribution of graph aggregation. The service patterns also reflect the propagation geometries. The Etoile deployment region is less obstructed, while the Munich UT cluster lies among denser buildings. These two layouts illustrate distinct relocation patterns under different urban geometries, and broader multi-seed cross-city evaluation remains future work.

\subsection{Learning-Curve Analysis}
Fig.~\ref{fig:ideal_training_curves} shows the learning curves in the ideal simulator. In Fig.~\ref{fig:ideal_training_curves}(a), GADC and GAT-MADDPG approach a high coverage region within the first several hundred episodes and plateau near coverage ratios of $0.727$ and $0.725$, respectively. MADDPG converges to a lower coverage ratio near $0.649$, and the GAT-MAPPO ablation plateaus near $0.559$. In Fig.~\ref{fig:ideal_training_curves}(b), the lifetime curves remain close to one for all methods after the early exploration stage. The curves match the ablation results. The comparison supports the empirical role of the coverage action-value update and graph representation in the tested setting.

{The curves show residual training fluctuations after episode $800$, during continued exploration. The exploration rate at episode $800$ is still $0.423$, and the GADC coverage ratio standard deviation after episode $800$ is about $0.040$ across seeds. Evaluation freezes the policy and selects actions greedily, removing training-time action randomization. Variability across layouts can remain.}

\subsection{Practical deployment considerations}
{
Physical deployment requires consideration of state estimation errors, control-message delivery, trajectory tracking, and onboard computation.
The tests in Fig.~\ref{fig:r1_robustness_suite} examine noisy position observations and delayed topology updates. 
The timing results in Table~\ref{tab:complexity_timing} show that action inference remains below $1.6$ ms for up to 50 UAVs on the tested workstation.
These
results characterize the tested sensing and topology perturbations and the computational cost on the workstation. 
Wireless delivery delays and execution time on embedded UAV hardware require separate evaluation.}

{
The five-action set introduces directional quantization. 
A diagonal relocation may require consecutive cardinal moves, while repeated direction changes can increase travel time and propulsion energy.
Reducing the movement increment improves spatial resolution but increases the number of decisions required to traverse a given distance. 
The movement increment and slot duration must therefore be selected according to the tracking limits of the flight controller.
A low-level trajectory generator can impose velocity, acceleration, and jerk limits on the commanded waypoints~\cite{px4jerk}, although the resulting physical trajectory may differ from the discrete simulated trajectory. 
Such differences can affect service links, displacement, and energy consumption. 
The ideal and ray-tracing simulators use discrete position updates and do not model acceleration, attitude dynamics, obstacle avoidance, or trajectory smoothing.
The ray-tracing evaluation characterizes radio propagation and does not establish flight-dynamic feasibility.}

{
GAT information exchange occurs at slot boundaries rather than continuously in physical time. 
Consider a practical implementation that performs the complete two-round embedding exchange once every $k$ decision slots and retains the received embeddings between exchange events. 
Let $E_n^{\mathrm{ex}}$ denote the incremental energy of one complete exchange event, including transmission and reception in both rounds. 
Let $E_n^{\mathrm{other}}$ denote the average energy per slot consumed by activities that continue between exchange events.
For fixed trajectories and service loads, if exchange-related radio activity is suspended between events, the average energy per slot is approximately
\begin{equation}
\overline{E}_n(k)
\approx
E_n^{\mathrm{other}}
+\frac{E_n^{\mathrm{ex}}}{k}.
\label{eq:exchange_interval_energy}
\end{equation}
Computation and idle listening that continue in every slot remain in $E_n^{\mathrm{other}}$. 
The effective coordination term in Eq.~\eqref{RESIDUALEE} does not separately measure these hardware-level components.}

{
A larger $k$ reduces the exchange contribution under the stated conditions, but the resulting lifetime improvement depends on the energy composition of the first UAV to exhaust its battery. 
Less frequent exchanges also increase the age of neighboring information and can change UAV trajectories and service loads.
These effects can offset part of the communication saving, so network lifetime need not increase monotonically with $k$. 
The experiment in Fig.~\ref{fig:r1_robustness_comm} changes the topology refresh interval while retaining current node features and the same per-slot communication-energy rule.
It therefore evaluates sensitivity to delayed topology updates rather than the energy saving from intermittent broadcasting. 
Its 60-slot evaluation horizon also does not quantify the lifetime benefit over longer missions.
}

\section{Conclusions}
\label{Conclusions}
This study focused on optimizing dual objectives in multi-UAV systems, maximizing user terminal (UT) coverage while extending UAV network lifetime. 
We developed the Graph Attention-based Decentralized Actor-Critic (GADC) method, which leverages Graph Attention Network layers to aggregate information from neighboring UAVs for decentralized coordination. 
{With dual critics, GADC uses the primary critic for coverage-oriented updates and an energy-reward advantage for replay-based clipped secondary refinement.}
We also tested the proposed scheme in NVIDIA Sionna, where the UAV-UT channels are generated by ray tracing.
{The numerical results show competitive coverage and lifetime performance. In the nominal-density sensitivity study, coverage varies less across the tested clipping parameters than across the tested weighted-sum coefficients.}

\ifCLASSOPTIONcaptionsoff
  \newpage
\fi

\bibliographystyle{IEEEtran}
\bibliography{Bibtex/ref}
\appendices
\setlength{\jot}{0.5pt}
\setlength{\abovedisplayskip}{1pt plus 0.5pt minus 0.5pt}
\setlength{\belowdisplayskip}{1pt plus 0.5pt minus 0.5pt}
\setlength{\abovedisplayshortskip}{0.5pt plus 0.3pt minus 0.3pt}
\setlength{\belowdisplayshortskip}{0.5pt plus 0.3pt minus 0.3pt}
\section{MATHEMATICAL DERIVATION OF THE PERFORMANCE BOUND}
\label{Proof}
\setcounter{lemma}{0}
\setcounter{theorem}{0}
\begin{lemma}[Conditional Coverage-Change Bound under a KL Neighborhood]
\label{lemma1_proof}
Let $\mathcal{J}^c(\pi)$ denote the expected discounted return for the
coverage objective. Given a baseline coverage policy $\pi^{c}$ and a
candidate energy-refined policy $\pi^{f}$ satisfying
\begin{equation}
\mathbb{E}_{s\sim d^{\pi^{c}}}
\left[
D_{KL}
\left(
\pi^{f}(\cdot|s)\,\|\,\pi^{c}(\cdot|s)
\right)
\right]
\leq \delta,
\end{equation}
the coverage-return difference satisfies
\begin{equation}
\mathcal{J}^{c}_{\pi^{f}}-\mathcal{J}^{c}_{\pi^{c}}
\geq
\frac{1}{1-\gamma}
\mathop{\mathbb{E}}_{\substack{s\sim d^{\pi^{c}}\\a\sim\pi^{f}}}
\left[A^{\pi^{c}}(s,a)\right]
-
\frac{\sqrt{2\delta}\gamma\eta_{\pi^{f}}}
{(1-\gamma)^2},
\label{eq:final_bound1}
\end{equation}
where $\delta$ specifies the KL-neighborhood bound and
$\eta_{\pi^{f}}
=
\max_s
\left|
\mathbb{E}_{a\sim\pi^{f}}
[A^{\pi^{c}}(s,a)]
\right|$.
\end{lemma}

\begin{proof}
We start with the fundamental relationship for the difference in expected returns between two policies $\pi^{f}$ and $\pi^{c}$, utilizing the \textbf{Performance Difference Lemma} \cite{10.5555/645531.656005}. 
Let $A^{\pi^{c}}(s, a) = Q^{\pi^{c}}(s,a) - V^{\pi^{c}}(s)$ be the advantage function of the baseline policy $\pi^{c}$. The performance difference is derived by expanding the trajectory expectation
\begin{equation}
    \begin{aligned}
        &\mathcal{J}^{c}_{\pi^{f}} - \mathcal{J}^{c}_{\pi^{c}} \\
        &= \mathbb{E}_{\tau \sim \pi^{f}} \left[ \sum_{t=0}^\infty \gamma^t r^c(s_t, a_t) \right] - \mathbb E_{s_0}[V^{\pi^{c}}(s_0)] \\
        &= \mathbb{E}_{\tau \sim \pi^{f}} \left[ \sum_{t=0}^\infty \gamma^t \left( r^c(s_t, a_t) + \gamma V^{\pi^{c}}(s_{t+1}) - V^{\pi^{c}}(s_t) \right) \right] \\
        &= \mathbb{E}_{\tau \sim \pi^{f}} \left[ \sum_{t=0}^\infty \gamma^t A^{\pi^{c}}(s_t, a_t) \right] \\
        &= \frac{1}{1-\gamma} \mathbb{E}_{s \sim d^{\pi^{f}}} \left[ \mathbb{E}_{a \sim \pi^{f}(\cdot|s)} [A^{\pi^{c}}(s, a)] \right].
    \end{aligned}
    \label{eq:pdl_result}
\end{equation}
This equation rigorously links the difference in cumulative returns to the expected advantage under the state distribution $d^{\pi^{f}}(s) = (1-\gamma)\sum_{t=0}^\infty \gamma^t P(s_t=s|\pi^{f})$.
Equation~\eqref{eq:pdl_result} relies on the state distribution $d^{\pi^{f}}$, which is unknown during optimization since it depends on the new policy $\pi^{f}$. To make this tractable, we decompose the expectation by introducing the baseline distribution $d^{\pi^{c}}$. Let $\bar{A}_{\pi^{f}}(s) = \mathbb{E}_{a \sim \pi^{f}(\cdot|s)} [A^{\pi^{c}}(s, a)]$. We rewrite Eq.~\eqref{eq:pdl_result} as
\begin{equation}
\begin{aligned}
&\mathcal{J}^{c}_{\pi^{f}} - \mathcal{J}^{c}_{\pi^{c}} = \frac{1}{1-\gamma} \sum_{s} d^{\pi^{f}}(s) \bar{A}_{\pi^{f}}(s) \\
&= \frac{1}{1-\gamma} \left( \sum_{s} d^{\pi^{c}}(s) \bar{A}_{\pi^{f}}(s) + \sum_{s} (d^{\pi^{f}}(s) - d^{\pi^{c}}(s)) \bar{A}_{\pi^{f}}(s) \right) \\
&= \underbrace{\frac{1}{1-\gamma} \mathop{\mathbb{E}}_{s \sim d^{\pi^{c}}} [\bar{A}_{\pi^{f}}(s)]}_{\text{Surrogate Objective}} + \underbrace{\frac{1}{1-\gamma} \langle d^{\pi^{f}} - d^{\pi^{c}}, \bar{A}_{\pi^{f}} \rangle}_{\text{Distribution Shift Error}}.
\end{aligned}
\label{eq:decomposition}
\end{equation}

The first term is the standard surrogate objective used in policy gradients. The second term represents the error introduced by the distribution shift. We bound this error term using \textbf{H\"{o}lder's Inequality}
\begin{equation}
\begin{aligned}
\left| \langle d^{\pi^{f}} - d^{\pi^{c}}, \bar{A}_{\pi^{f}} \rangle \right| 
&\le \| d^{\pi^{f}} - d^{\pi^{c}} \|_1 \cdot \| \bar{A}_{\pi^{f}} \|_\infty \\
&= 2 D_{TV}(d^{\pi^{f}} || d^{\pi^{c}}) \cdot \eta_{\pi^{f}},
\end{aligned}
\label{eq:holder}
\end{equation}
where $\eta_{\pi^{f}} = \max_s |\bar{A}_{\pi^{f}}(s)|$. 
According to the variational bound derived in CPO~\cite{pmlr-v70-achiam17a}, the divergence in state distributions is bounded by the divergence in policies,
\begin{equation}
D_{TV}(d^{\pi^{f}} || d^{\pi^{c}}) \le \frac{\gamma}{1-\gamma} \mathbb{E}_{s \sim d^{\pi^{c}}} [D_{TV}(\pi^{f}(\cdot|s) || \pi^{c}(\cdot|s))].
\end{equation}
Substituting this back into Eq.~\eqref{eq:holder}, and applying \textbf{Pinsker's Inequality} ($D_{TV} \le \sqrt{D_{KL}/2}$) together with Jensen's inequality gives
\begin{equation}
\text{Error} \le 2 \left( \frac{\gamma}{1-\gamma} \sqrt{\frac{1}{2} \delta} \right) \eta_{\pi^{f}} 
= \frac{\sqrt{2\delta}\gamma \eta_{\pi^{f}}}{1-\gamma}.
\end{equation}
Finally, substituting the error bound back into Eq.~\eqref{eq:decomposition} (which includes the $1/(1-\gamma)$ factor) yields the lower bound
\begin{equation}
\mathcal{J}^{c}_{\pi^{f}} - \mathcal{J}^{c}_{\pi^{c}} \geq \frac{1}{1-\gamma} \mathop{\mathbb{E}}_{\substack{s \sim d^{\pi^{c}} \\ a \sim \pi^{f}}} [A^{\pi^{c}}(s, a)] - \frac{\sqrt{2\delta}\gamma \eta_{\pi^{f}}}{(1-\gamma)^2}.
\end{equation}
This concludes the proof.
\end{proof}

\begingroup
\makeatletter
\def\@IEEEBIOskipN{4pt}
\makeatother
\linespread{0.85}\selectfont

\begin{IEEEbiography}[{\includegraphics[width=1in,height=1.25in,clip,keepaspectratio]{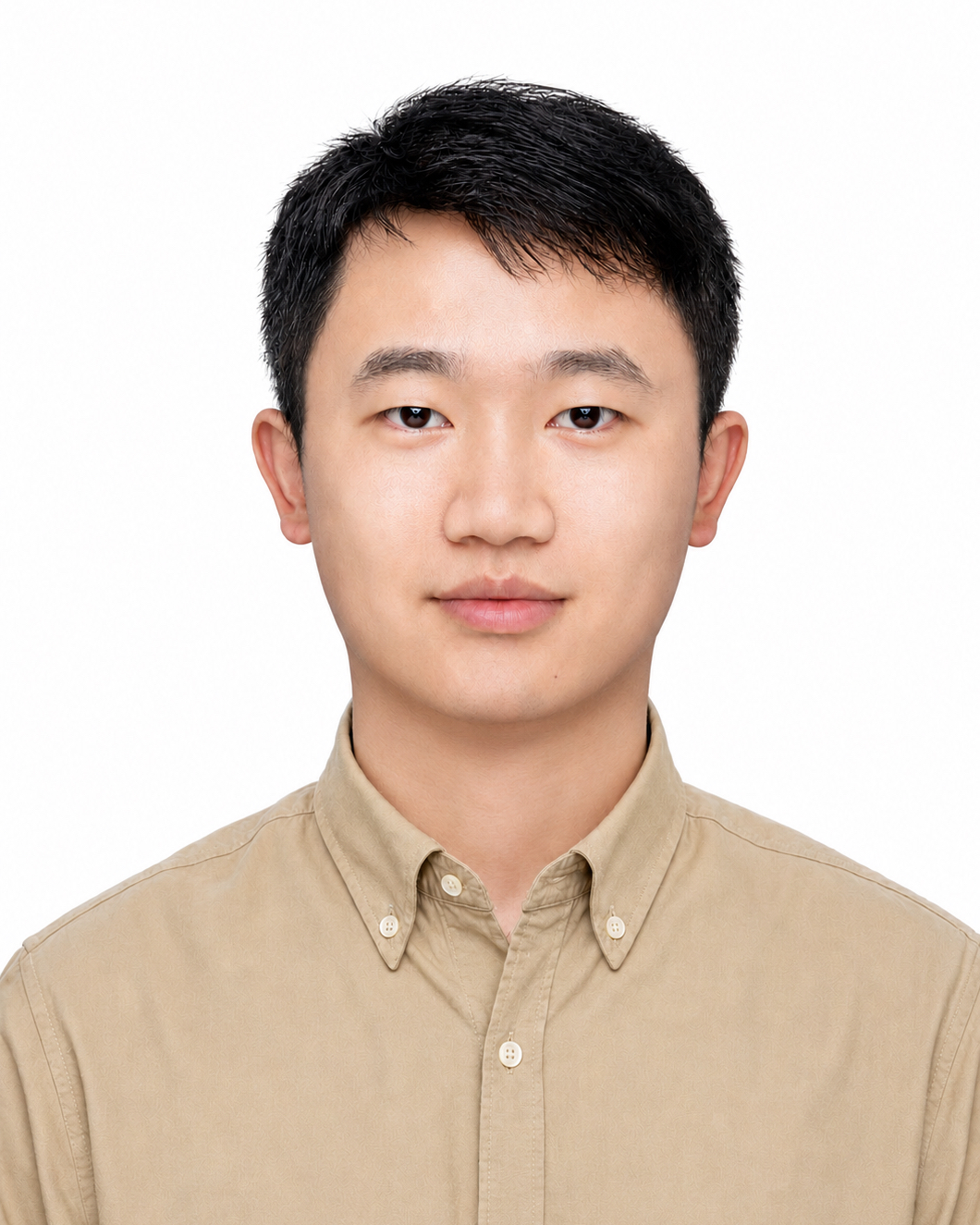}}]{Haoran Peng}
(Member, IEEE) received the B.Eng. degree in software engineering from the University of Electronic Science and Technology of China in 2015 and the Ph.D. degree (Hons.) in electrical and computer engineering from National Yang Ming Chiao Tung University in 2022.
He is currently a postdoctoral researcher with the Department of Information Engineering, The Chinese University of Hong Kong.
He previously held a postdoctoral appointment at the NYCU--IBM iIoT Research Center.
In 2021, he was a visiting student at Rochester Institute of Technology.
From 2015 to 2018, he worked at a software startup he co-founded.
He served as Publicity Chair of WOCC 2024 and as a Track Co-Chair at ICICN 2026.
\end{IEEEbiography}

\begin{IEEEbiography}[{\includegraphics[width=1in,height=1.25in,clip,keepaspectratio]{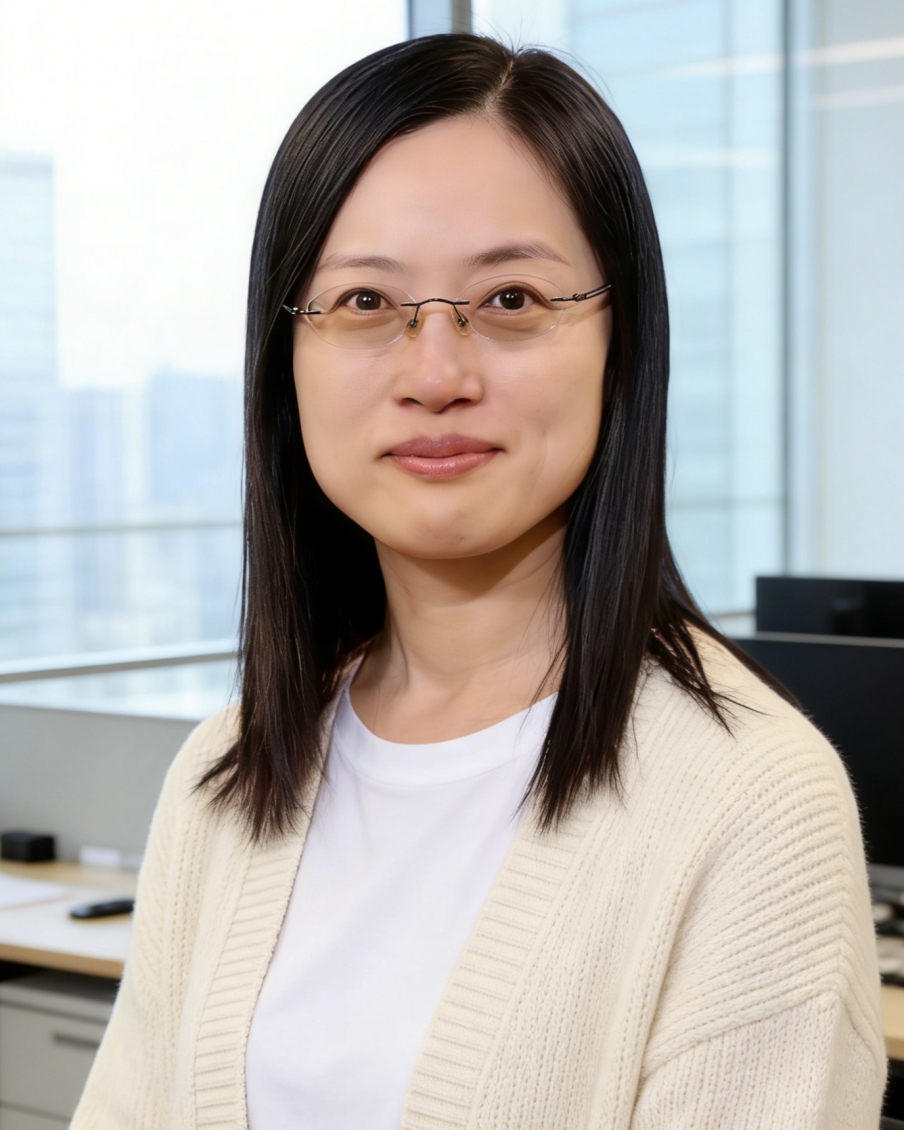}}]
{Ying-Jun Angela Zhang} (S’00-M’05’-SM10-F’20)  received her Ph.D. degree from the Department of Electrical and Electronic Engineering, The Hong Kong University of Science and Technology. She joined the Department of Information Engineering, The Chinese University of Hong Kong in 2005, where she is now a professor. 
Prof. Zhang is now the Steering Committee Chair of IEEE Wireless Communication Letters, Director of IEEE ComSoc Asia Pacific Region Board, and member of IEEE VTS Fellow Evaluation Standing Committee. Previously, she served as Editor-in-Chief of IEEE Open Journal of the Communications Society, Chair of the Executive Editor Committee of IEEE Transactions on Wireless Communications, and member of IEEE ComSoc Fellow Evaluation Standing Committee. Prof. Zhang has served on the Organizing Committees of many top conferences, such as IEEE GLOBECOM, ICC, VTC, SmartgridComm, and was the Chair of IEEE ComSoc Technical Committee of Smart Grid Communications.

Prof. Zhang is a co-recipient of 2021 and 2014 IEEE ComSoc Asia Pacific Outstanding Paper Awards, 2013 IEEE SmartgridComm Best Paper Award, and 2011 IEEE Marconi Prize Paper Award on Wireless Communications. As the only winner from engineering science, Prof. Zhang won the Hong Kong Young Scientist Award 2006, conferred by the Hong Kong Institute of Science. 
\end{IEEEbiography}
\endgroup
\end{document}